\documentclass[preprint]{elsarticle}

\usepackage{amsmath}
\usepackage{amsthm}
\usepackage{amssymb}
\usepackage{subfigure}
\usepackage{fullpage}
\usepackage{algorithm}
\usepackage{algpseudocode}
\usepackage{graphicx}
\usepackage{color}
\usepackage{hyperref}
\usepackage{gensymb}

\renewcommand{\algorithmicrequire}{\textbf{Initialization:}}

\newcommand{\removed}[1]{}

\newcommand{\rem}{\mathbf}

\newcommand{\ca}{\mathcal{A}}

\newcommand{\ra}{\rightarrow}
\newcommand{\rsa}{\rightsquigarrow}

\newcommand{\coleq}{\mathrel{\mathop:}=}
\newcommand{\bs}{\backslash}

\newcommand{\dprime}{{\prime\prime}}

\newcommand{\E}{\textup{E}}

\newtheorem{observation}{Observation}
\newtheorem{definition}{Definition}
\newtheorem{lemma}{Lemma}
\newtheorem{theorem}{Theorem}
\newtheorem{remark}{Remark}
\newtheorem{corollary}{Corollary}
\newtheorem{proposition}{Proposition}

\journal{}
\date{}

\begin{document}

\begin{frontmatter}



\title{Connectivity Preserving Network Transformers\tnoteref{t1}}

\tnotetext[t1]{Supported in part by (i) the project ``Foundations of Dynamic Distributed Computing Systems'' (\textsf{FOCUS}) which is implemented under the ``ARISTEIA'' Action of the  Operational Programme ``Education and Lifelong Learning'' and is co-funded by the European Union (European Social Fund) and Greek National Resources and (ii) the FET EU IP project \textsf{MULTIPLEX} under contract no 317532.}

\author[cti]{Othon Michail\corref{cor1}}
\ead{michailo@cti.gr}
\author[cti,liverpool]{Paul G. Spirakis}
\ead{P.Spirakis@liverpool.ac.uk}

\address[cti]{Computer Technology Institute and Press ``Diophantus'' (CTI), Patras, Greece}
\address[liverpool]{Department of Computer Science, University of Liverpool, Liverpool, UK}

\cortext[cor1]{Corresponding author (Telephone number: +30 2610 960300, Fax number: +30 2610 960490, Postal Address: Computer Technology Institute and Press ``Diophantus'' (CTI), N. Kazantzaki Str., Patras University Campus, Rion, P.O. Box 1382, 26504 Patras, Greece).}

\begin{abstract}
The Population Protocol model is a distributed model that concerns systems of very weak computational entities that cannot control the way they interact. The model of Network Constructors is a variant of Population Protocols capable of (algorithmically) constructing abstract networks. Both models are characterized by a \emph{fundamental inability to terminate}. In this work, we investigate the minimal strengthenings of the latter model that could overcome this inability. Our main conclusion is that \emph{initial connectivity of the communication topology} combined with the ability of the protocol to \emph{transform the communication topology} and the ability of a node to \emph{detect when its degree is equal to a small constant}, plus either a \emph{unique leader} or the ability of \emph{detecting common neighbors}, are sufficient to guarantee not only \emph{termination} but also the \emph{maximum computational power that one can hope for in this family of models}. The technique of our protocols is to transform any initial connected topology to a \emph{less symmetric} (that can serve as a large ordered memory) and \emph{detectable} (that the protocol can determine its successful construction) topology \emph{without ever breaking its connectivity} during the transformation. The target topology of all of our transformers is the \emph{spanning line} and we call \emph{Terminating Line Transformation} the corresponding problem. We first study the case in which there is a pre-elected \emph{unique leader} and give a \emph{time-optimal protocol for Terminating Line Transformation}. We then prove that dropping the leader without additional assumptions leads to a \emph{strong impossibility result}. In an attempt to overcome this, we equip the nodes with the ability to tell, during their pairwise interactions, whether they have at least one neighbor in common. Interestingly, it turns out that this local and realistic mechanism is sufficient to make the problem solvable. In particular, we give a \emph{very efficient protocol} that solves Terminating Line Transformation when all nodes are initially identical. The latter implies the aforementioned strong characterization: \emph{the model, under a few minimal assumptions, computes with termination any symmetric predicate computable by a Turing Machine of space $\Theta(n^2)$.}
\end{abstract}

\begin{keyword}

population protocol \sep network construction \sep network transformation \sep dynamic topology \sep termination \sep continuous connectivity \sep random scheduler \sep Turing Machine simulation
\end{keyword}

\end{frontmatter}

\section{Introduction}
\label{sec:intro}

A dynamic distributed computing system is a system composed of distributed computational processes in which the structure of the communication network between the processes changes over time. In one extreme, the processes cannot control and cannot accurately predict the modifications of the communication topology. Typical such examples are mobile distributed systems in which the mobility is \emph{external} to the processes and is usually provided by the environment in which the system operates. For example, it could be a system of cell phones following the movement of the individuals carrying them or a system of nanosensors flowing in the human circulatory system. This type of mobility is known as \emph{passive} (see e.g. \cite{AADFP06}). On the other extreme, dynamicity may be a sole outcome of the algorithm executed by the processes. Typical examples are systems in which the processes are equipped with some \emph{internal} mobility mechanism, like mobile robotic systems and, in general, any system with the ability to algorithmically modify the communication topology. This type of mobility is known as \emph{active} mobility (see e.g. \cite{WCG13} for active self-assembly, \cite{SY99,DFSY15,CKLL09} for mobile robots, and \cite{ACD11} for reconfigurable (nano)robotics under physical constraints). Recently, there is an interest in \emph{intermediate} (or \emph{hybrid}) systems. One such type, consists of systems in which the processes are passively mobile but still they are equipped with an internal active mechanism that allows them to have a partial (algorithmic) control of the system's dynamicity.

The intermediate model that guides our study here, is the \emph{network constructors} model introduced in \cite{MS14}. In this model, there are $n$ extremely weak processes, computationally equivalent to anonymous finite automata, that usually have very limited knowledge of the system (e.g. they do not know its size). The processes move passively and interact in pairs whenever two of them come sufficiently close to each other. This part of the system's dynamicity is not controlled and cannot be (completely) predicted by the processes and is modeled by assuming an adversary scheduler that in every step selects a pair of processes to interact. The adversary is typically restricted to be \emph{fair} so that it cannot forever block the system's progress (e.g. by keeping two parts of the system forever disconnected). Fairness is sufficient for analyzing the correctness of protocols for specific tasks. If additionally an estimate of the running time is desired, a typical assumption is that the scheduler is a uniform random one (which is fair with probability 1 \cite{CDFMS09} and also corresponds to the dynamicity patterns of well-mixed solutions). But in this model, there is also an internal source of dynamicity. In particular, the processes can algorithmically connect and disconnect to each other during their pairwise interactions. This can be viewed either as a physical bonding mechanism, as e.g. in reconfigurable robotics and molecular (e.g. DNA) self-assembly, or as a virtual record of local connectivity, as e.g. in a social network where a participant keeps track of and can regularly update the set of his/her associates. This allows the processes to control the construction and maintenance of a network or a shape in an uncontrolled and unpredictable dynamic environment. 

The network from which the scheduler picks interactions between processes and develops the uncontrolled interaction pattern is called the \emph{interaction network}. At the same time, the processes, by connecting and disconnecting to each other, develop another network, the \emph{(algorithmically) constructed network}, which is a subnetwork of the interaction network. In the most abstract setting, the interaction network is the clique $K_n$ throughout the execution, no matter what the protocol does (e.g. no matter how the protocol modifies the constructed network). In this case, the scheduler can in every step (throughout the course of the protocol) pick any possible pair of processes to interact, independently of the constructed network. \footnote{A convenient way to think of this setting is to imagine a clique graph with its edges labeled from $\{0,1\}$. Then, in this case, the clique is the interaction network while its subgraph induced by the edges labeled 1 is the constructed network.} This is precisely the setting of \cite{MS14} and also the one that we will consider in the present work. \footnote{On the other hand, it is possible, and plausible w.r.t. several application scenarios, that the set of available interactions at a given step actually depends on the constructed network. Such a case was considered in \cite{Mi15}, where the constructed network is always a subnetwork of the grid network and two processes can only interact if a connection between them would preserve this requirement. So, in that case, the set of available interactions is, in every step, constrained by the network that has been constructed by the protocol so far.} But even if the interaction network is always a clique independently of the constructed network, the ability of the processes to construct a network may still allow them to counterbalance the adversary's power. For example, if the processes manage somehow to self-organize into a spanning network $G$, then it might be possible for them to ignore all interactions that occur over the non-links of $G$ and thus force the actual communication pattern to be consistent with the constructed network.

The existing literature on distributed network construction \cite{MS14,Mi15} has almost absolutely focused on the setting in which all processes are initially disconnected and the goal is for them to algorithmically self-organize into a desired (usually spanning or of size at least some required function of $n$) stable network or shape. In \cite{MS14}, the authors presented simple and efficient direct constructors and lower bounds for several basic network construction problems such as spanning line, spanning ring, and spanning star and also generic constructors capable of constructing a large class of networks by simulating a Turing Machine (abbreviated ``TM'' throughout). One of the main results was that for every graph language $L$ that is decidable by a $O(\sqrt{l})$-space ($l+O(\sqrt{l})$, resp.) TM, where $l=\Theta(n^2)$ is the binary length of the input of the simulated TM, there is a protocol that constructs $L$ equiprobably with useful space $\lfloor n/2\rfloor$ ($\lfloor n/3\rfloor$, resp.), where the \emph{useful space} is defined as a lower bound on the order of the output network (the rest of the nodes being used as auxiliary and thrown away eventually as \emph{waste}). In \cite{Mi15}, a geometrically constrained variant was studied, where the formed network and the allowable interactions must respect the structure of the 2-dimensional (or 3-dimensional) grid network. The main result was a \emph{terminating} protocol counting the size $n$ of the system with high probability (abbreviated ``w.h.p.'' throughout). This protocol was then used as a subroutine of universal constructors, establishing that the nodes can self-assemble w.h.p. into arbitrarily complex shapes while still being capable to terminate after completing the construction.

\subsection{Our Approach and Contribution}

The main goal of this work is to investigate minimal strengthenings of the population protocol and network constructors models that can maximize their computational power, also rendering them capable to terminate. To this end, we consider (for the first time in network constructors) the case in which the initial configuration of the edges is not the one in which all edges are \emph{inactive} (i.e. those that are in state 0). In particular, we assume that the initial configuration of the edges can be any configuration in which the \emph{active} (i.e. those that are in state 1) edges form \emph{a connected graph spanning the set of processes}. \footnote{Active and inactive \emph{edges} are not to be confused with active and passive \emph{mobility}. An edge is said to be \emph{active} if its state is 1 and it is said to be \emph{inactive} if its state is 0.} The initial configuration of the nodes is either, as in \cite{MS14}, the one in which all nodes are initially in the same state, e.g. in an initial state $q_0$, or (whenever needed) the one in which all nodes begin from $q_0$ apart from a pre-elected unique leader that begins from a distinct initial leader-state $l_0$. This choice is motivated by the fact that without some sort of bounded initial disconnectivity we can only hope for global computations and constructions that are \emph{eventually stabilizing} (and not \emph{terminating}), \footnote{Informally, a computation or construction in this family of models is called \emph{eventually stabilizing} (or just \emph{stabilizing}) if, in a finite number of steps, the ``output'' stops changing. In case of a typical computation, the  ``output'' usually concerns an output component in the states of the nodes and in case of a network constructor, it concerns the network induced by the active edges. A computation or construction is called \emph{terminating} if, in a finite number of steps, every node is in a \emph{halting state} (which is a state that does not have an effective interaction with any other state; that is, any interaction of such a state cannot result in a state update).} because a component can guess neither the number of components not encountered yet nor an upper bound on the time needed to interact with another one of them \footnote{Observe that this argument does not apply if the number of components is upper bounded by a constant $c$. In the latter case, the processes could know $c$ in advance and it might be possible for a process to determine, for example, when it has managed to connect its own component to all the other components. Still, for simplicity and clarity of presentation, we restrict our attention to connected initial configurations.}  (\cite{MS15} overcomes this by assuming that the nodes know some upper bound on this time, while \cite{Mi15} overcomes this by assuming a uniform random scheduler and a unique leader and by restricting correctness to be w.h.p.). 

Next, observe that if the protocol is not allowed to modify the state of the edges, then the assumption of initial connectivity alone does not add any computational power to the model (in the worst case). For if we ignore for a while the ability of the model to modify the state of the edges, what we have is a model equivalent to classical population protocols \cite{AADFP06} on a restricted interaction graph \cite{AACFJP05} (observe that the model can ignore the interactions that occur over inactive edges). Though there are some restricted interaction graphs, like the spanning line, that dramatically increase the computational power of the model (in this case making it equivalent to a TM of linear space), still there others, like the spanning star, on which the power of the model is as low as the power of classical population protocols on a clique interaction graph \cite{CMNS13}, which, in turn, is equal to the rather small class of \emph{semilinear predicates} \cite{AAER07}. As we have allowed any possible connected initial set of active edges, the spanning star inclusive, the initial configuration of the edges alone (without any edge modifications) is not sufficient for strengthening the model.

Our discussion so far, suggests to consider at the same time initial connectivity (or, more generally, bounded initial disconnectivity) and the ability of the protocol to modify the state of the edges, with the hope of increasing the computational power. Unfortunately, even with this additional assumption, non-trivial terminating computation is still impossible (this is proved in Proposition \ref{pro:termination-impossibility}, in Section \ref{subsec:fi}). An immediate way to appreciate this, is to notice that a clique does not provide more information than an empty network about the size of the system. Even worse, if a node's initial active degree is unbounded (as e.g. is the case for the center of a spanning star), then it is not clear even whether the stabilizing constructors that assume initial disconnectivity (as in \cite{MS14}) can be adapted to work. Actually, it could be the case, that \emph{without additional assumptions} initial connectivity may even decrease the power of the model (we leave this as an interesting open problem). For example, it could be simpler to construct a spanning line if the initial active network is empty (i.e. all edges are inactive) than if it is a clique (i.e. all edges are active). Even if it would turn out that the model does not become any weaker, we still cannot avoid the aforementioned impossibility of termination and the maximum that we can hope for is an \emph{eventually stabilizing} universal constructor, as the one of \cite{MS14}. 

We now add to the picture a very minimal and natural, but extremely powerful, additional assumption that, combined with our assumptions so far, will lead us to a stronger model. In particular, we equip the nodes with the ability to detect some small local degrees. For a concrete example, assume that a node can detect when its active degree is equal to 0 (otherwise it only knows that its degree is at least 1). A first immediate gain, is that we can now directly simulate any constructor that assumes an empty initial network (e.g. the constructors of \cite{MS14}): every node initially deactivates the active edges incident to it until its local active degree becomes for the first time 0, and only when this occurs the node starts participating in the simulation. So, even though a node does not know its initial degree (which is due to the fact that a node in this model is a finite automaton with a state whose size is independent of the size of the system), it can still detect when it becomes equal to 0. At that point, the node does not have any active edges incident to it, therefore it can start executing the constructor that assumes an empty initial network.

Our main finding in this work, is that the initial connectivity guarantee together with the ability to modify the network and to detect small local degrees (combined with either a pre-elected leader or a natural mechanism that allows two nodes to tell whether they have a neighbor in common), are sufficient to obtain the \emph{maximum computational power} that one can hope for in this family of models. In particular, the resulting model can compute \emph{with termination} any symmetric predicate \footnote{Essentially, a predicate in this type of models is called \emph{symmetric} (or \emph{commutative}) if permuting the input symbols does not affect the predicate's outcome. Such a restriction is imposed by the fact that, in general, the nodes cannot be distinguished initially, because they are anonymous and some initial networks are very symmetric. In such cases, there is no way for a protocol to distinguish one permutation of the the nodes' inputs from another.} computable by a \emph{TM of space $\Theta(n^2)$}, and no more than this, i.e. it is an exact characterization. The symmetricity restriction can only be dropped by UIDs or by any other means of knowing and maintaining an ordering of the nodes' inputs. This power is maximal because the distributed space of the system is $\Theta(n^2)$, so we cannot hope for computations exploiting more space. The substantial improvement compared to \cite{MS14,MCS11-2} is that the universal computations are now \emph{terminating} and not just \emph{eventually stabilizing}. It is interesting to point out that the additional assumptions and mechanisms are minimal, in the sense that the removal of each one of them leads to either an impossibility of termination or to a substantial decrease in the computational power.

Our approach to arriving at the above characterization is the following. We develop protocols that exploit the knowledge of the initial connectivity of the active topology \footnote{We use the terms \emph{network} and \emph{topology} interchangeably in this work.} and try to transform it to a less symmetric and detectable active topology without ever breaking its connectivity. \footnote{The \emph{symmetry} of an interaction network or a subnetwork of it, refers in this work to the memory capacity of the network under this particular model. For example, a clique is maximally symmetric, because it provides no ordering on the nodes (which are actually the cells of the distributed memory) and can only be exploited as a logarithmic distributed memory ($n$ cells in unary), while, on the other hand, a line is minimally symmetric, because it provides a total ordering on the nodes and can be exploited as a linear memory ($n$ cells in binary).} The knowledge of initial connectivity and its preservation throughout the transformation process, ensure that the protocol always has all nodes of the network in a single component. Still, if the component is very symmetric, e.g. if it is a clique or a star, then, as already discussed above, it cannot be exploited for powerful computations. Moreover, if the target-network is symmetric, then there might be no way for the transformation to determine when has managed to form the network. For example, without additional non-local assumptions, as is the unrealistic assumption of knowing that the inactive degree is 0, it is impossible for a node to determine that its degree is $n$, because a node can only count up to constants, so it is in general impossible to determine the formation of a star with it at the center. Instead, our protocols \emph{transform any spanning connected initial topology into a spanning line while preserving connectivity throughout the transformation process}. The spanning line has the advantage that it can be detected under the minimal assumption that a node can detect whether its local degree is in $\{1,2\}$ and that it is minimally symmetric and, therefore, capable of serving as a linear memory. Preservation of connectivity allows the protocol to be certain that the spanning line contains all processes. So, the protocol can detect the formation of the spanning line and then count (on $O(\log n)$ cells, i.e. nodes, of the linear distributed memory) the size of the system. Then the protocol can use the spanning line as it is, for simulating (on the nodes of the line) TMs of space $\Theta(n)$. Going one step further, it is not hard for a protocol to exploit all this obtained information and perform a final transformation that increases the simulation space to $\Theta(n^2)$ (in the spirit of some constructions of \cite{MS14}).

In particular, given an initially connected active topology and the ability of the protocol to transform the topology, we prove the following set of results:
\begin{itemize}
\item If there is a unique leader and a node can detect whether its degree is equal to 1, then there is a time-optimal protocol, with running time $\Theta(n^2\log n)$, \footnote{In this work, the \emph{running time} is studied under a uniform random scheduler and is defined as the maximum/worst-case expected running time over all possible initial active topologies.} that transforms any initial active topology to a spanning line and terminates.
This implies a full-power TM simulation as described above.
\item If all nodes are initially identical (and even if small local degrees can be detected) then there is no protocol that can transform any initial active topology to an acyclic topology without ever breaking connectivity. The impossibility result is quite strong, proving that for any initial topology $G$ there is an infinite family $\cal{G}$ such that if the protocol makes $G$ acyclic then it disconnects every $G^\prime\in \cal{G}$ in $\Theta(|V(G^\prime)|)$ parts. The latter implies that it is impossible to transform to a spanning line with termination.
\item There is a plausible additional strengthening that allows the problem to become solvable with initially identical nodes. In particular, we assume that when two nodes interact they can tell whether they have a neighbor in common (the corresponding mechanism called \emph{common neighbor detection}). This is a very local and thus plausible mechanism to assume. We prove that with this additional assumption, initially identical nodes can transform any connected spanning initial active topology to a spanning line and terminate in time $O(n^3)$. This implies a full-power TM simulation as described above.
\end{itemize}

In Section \ref{sec:rw}, we discuss further related literature. Section \ref{sec:prel} brings together all definitions and basic facts that are used throughout the paper. In particular, in Section \ref{subsec:model} we formally define the model of network constructors under consideration, Section \ref{subsec:problems} formally defines the transformation problems that are considered in this work, and Section \ref{subsec:fi} provides some basic impossibility results and a lower bound on the time needed to transform any network to a spanning line. In Section \ref{sec:unique-leader}, we study the case in which there is a pre-elected unique leader and give two protocols for the problem, the Online-Cycle-Elimination protocol (Section \ref{subsec:cycle-elimination}) and the time-optimal Line-Around-a-Star protocol (Section \ref{subsec:line-around-a-star}). Then, in Section \ref{sec:identical-nodes}, we try to drop the unique leader assumption. First, in Section \ref{subsec:impossibilities} we show that, without additional assumptions, dropping the unique leader leads to a strong impossibility result. In face of this negative result, in Section \ref{subsec:neighbor-detection} we minimally strengthen the model with a common neighbor detection mechanism and give a correct terminating protocol. Finally, in Section \ref{sec:conclusions} we conclude and give further research directions that are opened by our work.  

\section{Further Related Work}
\label{sec:rw}

The model considered in this paper belongs to the family of population protocol models. The population protocol model \cite{AADFP06} was originally developed as a model of highly dynamic networks of simple sensor nodes that cannot control their mobility. The first papers focused on the computational capabilities of the model which have now been almost completely characterized. In particular, if the interaction network is complete, i.e. one in which every pair of processes may interact, then the computational power of the model is equal to the class of the \emph{semilinear predicates} (and the same holds for several variations) \cite{AAER07}. Semilinearity persists up to $o(\log\log n)$ local space but not more than this \cite{MNPS11}. If additionally the connections between processes can hold a state from a finite domain (note that this is a stronger requirement than the active/inactive that the present work assumes) then the computational power dramatically increases to the commutative subclass of $\rem{NSPACE}(n^2)$ \cite{MCS11-2}. The latter constitutes the mediated population protocol (MPP) model, which was the first variant of population protocols to allow for states on the edges. For introductory texts to these models, the interested reader is encouraged to consult \cite{AR09} and \cite{MCS11}. 

Based on the MPP model, \cite{MS14} restricted attention to binary edge states and regarded them as a physical (or virtual, depending on the application) bonding mechanism. This gave rise to a ``hybrid'' self-assembly model, the network constructors model, in which the actual dynamicity is passive and due to the environment but still the protocol can construct a desired network by activating and deactivating appropriately the connections between the nodes. The present paper essentially investigates the computational power of the network constructors model under the assumption that a connected spanning active topology is provided initially and also initiates the study of the \emph{distributed network reconfiguration problem}. Recently, \cite{Mi15} studied a geometrically constrained variant of network constructors in which the interaction network is not complete but rather it is constrained by the existing shapes (every shape that can be formed being a sub-network of the 2D or 3D grid network). Interestingly, apart from being a model of computation, population protocols are also closely related to chemical systems. In particular, Doty \cite{Do14} has recently demonstrated their formal equivalence to \emph{chemical reaction networks} (CRNs), which model chemistry in a \emph{well-mixed solution}.

There are already several models trying to capture the self-assembly capability of natural processes with the purpose of engineering systems and developing algorithms inspired by such processes. For example, \cite{Do12} proposes to learn how to program molecules to manipulate themselves, grow into machines and at the same time control their own growth. The research area of ``algorithmic self-assembly'' belongs to the field of ``molecular computing''. The latter was initiated by Adleman \cite{Ad94}, who designed interacting DNA molecules to solve an instance of the Hamiltonian path problem. The model guiding the study in algorithmic self-assembly is the Abstract Tile Assembly Model (aTAM) \cite{Wi98,RW00} and variations (e.g. \cite{FHPR15}). In contrast to those models that try to incorporate the exact molecular mechanisms (like e.g. temperature, energy, and bounded degree), the abstract combinatorial rule-based model considered here is free of specific application-driven assumptions with the aim of revealing the fundamental laws governing the distributed (algorithmic) generation of networks. The present model may serve as a common substructure to more applied models (like assembly models or models with geometry restrictions, as in \cite{Mi15}) that may be obtained from it by imposing restrictions on the scheduler, the degree, and the number of local states.

Finally, we should mention that our motivating question here is ``how can we make a distributed system that consists of very weak entities and whose initial topology may be highly symmetric or without \emph{much structure}, not allowing interesting global computations, exploit its full potential computational capabilities?''. Though we pick population protocols as a starting point, we believe that this question (and also the problem of transforming between topologies/networks \emph{per se}) is important in several other contexts like robotics \cite{SY99,DFSY15,CKLL09,FYKY12}, self-configurable, self-assembly, programmable matter \cite{Mi15,DDGRS14}, and dynamic distributed systems \cite{OW05,KLO10,MCS14,MCS13b}. Imagine this question in a system of identical mobile robots arranged in the plane in a ring topology. This symmetric arrangement may not allow them to obtain global information or perform global computations. Moreover, trying to break the ring at some point (as all robots are identical) may disconnect the topology. Still there might be an algorithmic transformation to a more ``informative/useful'' topology (that allows for more powerful global computations) without ever disconnecting the system. Such a transformation is a weaker requirement than, for example, gathering \cite{CFPS03} and may be simpler to obtain.

\section{Preliminaries}
\label{sec:prel}

\subsection{The Model}
\label{subsec:model}

The model under consideration is the network constructors model of \cite{MS14} with the only essential difference being that in \cite{MS14} the initial configuration was always (apart from the network replication problem) the one in which all edges are inactive, while in this work the initial configuration can be any configuration in which the active edges form a spanning connected network. Still, we give a detailed presentation of the model for self-containment.

\begin{definition}
A \emph{Network Constructor} (NET) is a distributed protocol defined by a 4-tuple $(Q,q_0,Q_{out},\delta)$, where
$Q$ is a finite set of \emph{node-states}, $q_0\in Q$ is the \emph{initial node-state}, $Q_{out}\subseteq Q$ is the set of \emph{output node-states}, and $\delta : Q\times Q\times \{0,1\} \rightarrow Q\times Q\times \{0,1\}$ is the \emph{transition function}. When required, also a special \emph{initial leader-state} $l_0\in Q$ may be defined.
\end{definition}

If $\delta(a,b,c) = (a^{\prime},b^{\prime},c^{\prime})$, we call $(a,b,c) \rightarrow (a^{\prime},b^{\prime},c^{\prime})$ a \emph{transition} (or \emph{rule}) and we define $\delta_{1}(a,b,c) = a^{\prime}$, $\delta_{2}(a,b,c) = b^{\prime}$, and $\delta_{3}(a,b,c) = c^{\prime}$. A transition $(a,b,c) \rightarrow (a^{\prime},b^{\prime},c^{\prime})$ is called \emph{effective} if $x\neq x^\prime$ for at least one $x\in\{a,b,c\}$ and \emph{ineffective} otherwise. When we present the transition function of a protocol we only present the effective transitions. Additionally, we agree that the \emph{size} of a protocol is the number of its states, i.e. $|Q|$.

The system consists of a population $V_I$ of $n$ distributed \emph{processes} (also called \emph{nodes} when clear from context). In the generic case, there is an underlying \emph{interaction graph} $G_I=(V_I,E_I)$ specifying the permissible interactions between the nodes. Interactions in this model are always pairwise. In this work, unless otherwise stated, $G_I$ is a \emph{complete undirected interaction graph}, i.e. $E_I=\{uv:u,v\in V_I \text{ and } u\neq v\}$, where $uv=\{u,v\}$. When we say that all nodes in $V_I$ are initially \emph{identical}, we mean that all nodes begin from the initial node-state $q_0$. In case we assume the existence of a unique leader, then there is a $u\in V_I$ beginning from the initial leader-state $l_0$ and all other $v\in V_I\bs\{u\}$ begin from the initial node-state $q_0$ (which in this case may also be called the \emph{initial nonleader-state}).

A central assumption of the model is that edges have binary states. An edge in state 0 is said to be \emph{inactive} while an edge in state 1 is said to be \emph{active}. In almost all problems studied in \cite{MS14} (apart from the replication problem), all edges were initially inactive. Though we shall also consider this case in the present paper, our main focus is on a different setting in which the protocol begins its execution on a precomputed set of active edges provided by some adversary. Formally, there is an input set of edges $E\subseteq E_I$, such that all $e\in E$ are initially active and all $e^\prime\in E_I\bs E$ are initially inactive. The set $E$ defines the \emph{input graph} $G=(V_I,E)$, also called the \emph{initial active topology/graph}. Throughout this work, unless otherwise stated, we assume that the initial active topology is \emph{connected}, which means that the active edges form a connected graph spanning $V_I$. This is a restriction imposed on the adversary selecting the input. In particular, the adversary is allowed to choose any initial set of active edges $E$ (in a worst-case manner), subject to the constraint that $E$ defines a connected graph on the whole population. 

Execution of the protocol proceeds in discrete steps. In every step, a pair of nodes $uv$ from $E_I$ is selected by an \emph{adversary scheduler} and these nodes interact and update their states and the state of the edge joining them according to the transition function $\delta$. In particular, we assume that, for all distinct node-states $a,b\in Q$ and for all edge-states $c\in\{0,1\}$, $\delta$ specifies either $(a,b,c)$ or $(b,a,c)$. So, if $a$, $b$, and $c$ are the states of nodes $u$, $v$, and edge $uv$, respectively, then the unique rule corresponding to these states, let it be $(a,b,c)\rightarrow (a^\prime,b^\prime,c^\prime)$, is applied, the edge that was in state $c$ updates its state to $c^\prime$ and if $a\neq b$, then $u$ updates its state to $a^\prime$ and $v$ updates its state to $b^\prime$, if $a=b$ and $a^\prime=b^\prime$, then both nodes update their states to $a^\prime$, and if $a=b$ and $a^\prime\neq b^\prime$, then the node that gets $a^\prime$ is drawn equiprobably from the two interacting nodes and the other node gets $b^\prime$. 

A \emph{configuration} is a mapping $C : V_I\cup E_I \rightarrow Q\cup \{0,1\}$ specifying the state of each node and each edge of the interaction graph. Let $C$ and $C^{\prime}$ be configurations, and let $u$, $\upsilon$ be distinct nodes. We say that \emph{$C$ goes to $C^{\prime}$ via encounter $e=u\upsilon$}, denoted $C \stackrel{e}\rightarrow C^{\prime}$, if $(C^{\prime}(u),C^{\prime}(v),C^{\prime}(e))=$ $\delta(C(u),C(v),C(e)) \text{ or }$ $(C^{\prime}(v),C^{\prime}(u),C^{\prime}(e))=$ $\delta(C(v),C(u),C(e)) \text{ and }$ $C^{\prime}(z)=$ $C(z), \mbox{ for all } z\in (V_I\bs\{u,v\})\cup (E_I\bs\{e\})$. We say that \emph{$C^\prime$ is reachable in one step from $C$}, denoted $C\rightarrow C^{\prime}$, if $C \stackrel{e}\rightarrow C^{\prime}$ for some encounter $e\in E_I$. We say that $C^{\prime}$ is \emph{reachable} from $C$ and write $C\rsa C^{\prime}$, if there is a sequence of configurations $C=C_{0},C_{1},\ldots,C_{t}=C^{\prime}$, such that $C_{i}\rightarrow C_{i+1}$ for all $i$, $0\leq i <t$.

An \emph{execution} is a finite or infinite sequence of configurations $C_{0},C_{1},$ $C_{2},\ldots$, where $C_{0}$ is an initial configuration and $C_{i}\rightarrow C_{i+1}$, for all $i\geq 0$. A \emph{fairness condition} is imposed on the adversary to ensure the protocol makes progress. An infinite execution is \emph{fair} if for every pair of configurations $C$ and $C^{\prime}$ such that $C\rightarrow C^{\prime}$, if $C$ occurs infinitely often in the execution then so does $C^{\prime}$. In what follows, every execution of a NET will by definition considered to be fair.

We define the \emph{output of a configuration} $C$ as the graph $G(C)=(V,E)$ where $V=\{u\in V_I: C(u)\in Q_{out}\}$ and $E=\{uv:u,v\in V,$ $u\neq v$, and $C(uv)=1\}$. In words, the output-graph of a configuration consists of those nodes that are in output states and those edges between them that are active, i.e. the active subgraph induced by the nodes that are in output states. The output of an execution $C_0,C_1,\ldots$ is said to \emph{stabilize} (or \emph{converge}) to a graph $G$ if there exists some step $t\geq 0$ s.t. $G(C_i)=G$ for all $i\geq t$, i.e. from step $t$ and onwards the output-graph remains unchanged. Every such configuration $C_i$, for $i\geq t$, is called \emph{output-stable}. The \emph{running time} (or \emph{time to convergence}) of an execution is defined as the minimum such $t$ (or $\infty$ if no such $t$ exists). Throughout the paper, whenever we study the running time of a NET, we assume that interactions are chosen by a \emph{uniform random scheduler} which, in every step, selects independently and uniformly at random one of the $|E_I|=n(n-1)/2$ possible interactions. In this case, the running time on a particular $n$ and an initial set of active edges $E$ becomes a random variable (abbreviated ``r.v.'') $X_{n,E}$ and our goal is to obtain bounds on $\max_{n,E}\{\E[X_{n,E}]\}$, where $\E[X]$ is the \emph{expectation} of the r.v. $X$. That is, the running time of a protocol is defined here as the maximum (also called \emph{worst-case}) expected running time over all possible initial configurations. Note that the uniform random scheduler is fair with probability 1.

\begin{definition}
We say that an execution of a NET on $n$ processes \emph{constructs a graph} (or \emph{network}) $G$, if its output stabilizes to a graph isomorphic to $G$.
\end{definition}

\begin{definition}
We say that a NET $\ca$ \emph{constructs a graph language $L$ with useful space $g(n)\leq n$}, if $g(n)$ is the greatest function for which: (i) for all $n$, every execution of $\ca$ on $n$ processes constructs a $G\in L$ of order at least $g(n)$ (provided that such a $G$ exists) and, additionally, (ii) for all $G\in L$ there is an execution of $\ca$ on $n$ processes, for some $n$ satisfying $|V(G)|\geq g(n)$, that constructs $G$. Equivalently, we say that \emph{$\ca$ constructs $L$ with waste $n-g(n)$}.
\end{definition}

In this work, we shall also be interested in NETs that construct a graph language and additionally always \emph{terminate}. 

\begin{definition}
We call a NET $\ca$ \emph{terminating} (or say that $\ca$ \emph{always terminates}) if every execution of $\ca$ reaches a \emph{halting} configuration, that is one in which every node is in a state $q_{h}$ from a set of halting states $Q_{halt}$, where $(q_h,q,s)\ra (q_h,q,s)$ (i.e. is ineffective) for every $q_h\in Q_{halt}$, $q\in Q$, and $s\in \{0,1\}$. 
\end{definition}

Observe that termination (as usual) is a property of the protocol and is by no means related to the scheduler, in the sense that the scheduler keeps selecting interactions forever but it happens that the protocol has entered a halting configuration in which no further state updates are possible no matter what the scheduler chooses to do. Additionally the protocol \emph{knows} that it has halted. This is in contrast to protocols with \emph{stabilizing states} \cite{CMNS13} in which, though the states eventually stabilize, the protocol cannot tell (e.g. announce) when this is the case because it does not use a distinguished set of halting states.

Finally, in order to consider TM simulations, we denote by $\rem{SSPACE}(f(n))$ the symmetric subclass of the complexity class $\rem{SPACE}(f(n))$ \footnote{$\rem{SPACE}(f(n))$ is the class of all languages decided by a deterministic TM using space $O(f(n))$.}, that is the one consisting of all symmetric languages in $\rem{SPACE}(f(n))$, where a language $L$ is \emph{symmetric} if for each input string $x\in L$, any permutation of $x$'s symbols also belongs to $L$. Note that symmetric languages are also known as \emph{commutative} \cite{FMR68}. The reason for restricting attention to symmetric languages is that the nodes are initially anonymous, therefore, in many initial configurations, a protocol has no means of distinguishing some permutations of the input assignment (in the extreme cases of a clique or an empty input graph, a protocol cannot distinguish any possible pair of permutations).   

\subsection{Problem Definitions}
\label{subsec:problems}

\noindent\textbf{Acyclicity.} Let $G=(V,A)$ be the subgraph of $G_I$ consisting of $V$ and the active edges between nodes in $V$, that is $A=\{e\in E_I: C(e)=1\}$. The initial $G$ is connected. The goal is for the processes to stably transform $G$ to an acyclic graph spanning $V$ without ever breaking the connectivity of $G$.\\

\noindent\textbf{Line Transformation.} Let $G=(V,A)$ be the subgraph of $G_I$ consisting of $V$ and the active edges between nodes in $V$, that is $A=\{e\in E_I: C(e)=1\}$. The initial $G$ is connected. The goal is for the processes to stably transform $G$ to a spanning line.\\

\noindent\textbf{Terminating Line Transformation.} The same as Line Transformation with the additional requirement that all processes must terminate.\\

Keep in mind that the Acyclicity problem requires the preservation of $G$'s connectivity while, on the other hand, Line Transformation does not. This is quite plausible, as making a connected graph acyclic without worrying about connectivity preservation is trivial (the protocol just has to eliminate all edges eventually).

\subsection{Fundamental Inabilities}
\label{subsec:fi}

We now give a few basic impossibility results that justify the necessity of minimally strengthening the network constructors model in order to be able to solve the above main problems.  

The following proposition (which is a well-known fact in the relevant literature but we include here a proof for self-containment) states that if the system does not involve edge states (i.e. the original population protocol model with transition function $\delta:Q\times Q\ra Q\times Q$), then a protocol cannot decide with termination whether there is a single $a$ in the population (mainly because a node does not know how much time it has to wait to meet every other node). Though the result is not directly applicable to our model, still we believe that it might help the reader's intuition w.r.t. to the computational difficulties in this family of models.

\begin{proposition} [PPs Impossibility of Termination] \label{pro:pp-termination-impossibility}
There is no population protocol that can compute with termination the predicate $(N_a\geq 1)$ (i.e. whether there exists an $a$ in the input assignment).
\end{proposition}
\begin{proof}
Consider a population of size $n$ and let the nodes be $u_1,u_2,\ldots,u_n$. It suffices to prove the impossibility for the variation in which there is a unique leader, initially in state $l$, and all other nodes are non-leaders, initially in state $q_a$ if their input is $a$ and $q_b$ if their input is not $a$, and all interactions are between the leader and the non-leaders. This is w.l.o.g. because this model is not weaker than the original population protocol model, which means that an impossibility for this model also transfers to the original population protocol model. Indeed, this model can easily simulate the original model as follows. Interactions between two non-leaders can be simulated via the leader: the leader first collects the state $q_1$ of a node $u$, which it marks, and then waits to interact with another node $v$. When this occurs, if the state of $v$ is $q_2$, rule $(q_1,q_2)\ra (q_1^\prime,q_2^\prime)$ is applied to the state stored by the leader and to the state of $v$. Then the leader waits to meet $u$ again (which can be detected since it has been marked) in order to update its state to $q_1^\prime$. When this occurs, the leader drops the stored information and starts a new simulation round.

So, let $u_1$ be the initial leader. Let $A$ be a protocol that computes $(N_a\geq 1)$ and terminates on every $n$ and every input assignment. Consider now the input assignment in which all inputs are $b$, that is there is no $a$ and thus all non-leaders begin with initial state $q_b$. Clearly, it must hold that in every fair execution the leader terminates in a finite number of steps and says ``no''. These steps are interactions between the leader and the non-leaders so any such execution can be represented by a sequence of $u_j$s from $\{u_2,\ldots,u_n\}$. Let now $s=v_1,v_2,\ldots,v_k$ be any such finite execution in which the leader says ``no''. $v_i\in \{u_2,\ldots,u_n\}$ is simply the node with which the leader interacted at step $i$.

Consider now a population of size $n+1$. The only difference to the previous setting is that now we have added a node $u_{n+1}$ with input $a$. Since now the predicate evaluates to 1, in every fair execution, $A$ should terminate in a finite number of steps and say ``yes''. Take any fair execution $s^\prime=s,v_{k+1},\ldots,v_h$, that is $s^\prime$ has $s$ as an ``unfair'' prefix. As $s$ contains the same nodes as before with the same input assignment, the leader in $s^\prime$ terminates in precisely $k$ steps saying ``no'' without knowing that an additional node with input $a$ exists in this case. This contradicts the existence of protocol $A$. We should mention that the leader has no means of guessing the existence of node $u_{n+1}$ because its termination only depends on the protocol stored in its memory which is by definition finite and independent of $n$ (suffices to consider the longest chain of rules that leads to termination with output ``no'' and which corresponds to at least one feasible execution).
\end{proof}

Moreover, even if the system is initially connected, there are some very symmetric topologies that do not allow for strong computations. For example, if the topology is a star with the leader at the center, then the system is equivalent to population protocols on a complete interaction graph and can compute only semilinear predicates on input assignments, again only in an eventually-stabilizing way (i.e. no termination). This is captured by the following proposition.

\begin{proposition} [Structure vs Computational Power]
There are initial topologies in which the computational power of population protocols is as limited as in the case of no structure at all.
\end{proposition}
\begin{proof}
Let the initial active topology be a star. Then it has been proved in \cite{CMNS13} that the stably computable predicates in this case are precisely the semilinear predicates, i.e. equal to the set of predicates computable in the standard population protocol model in which the communication topology is always a clique. 
\end{proof}

The above expose the necessity of additional assumptions, such as topology modifications, in order to hope for terminating computations and surpass the computational power of classical population protocols. So, we turn our attention again to our model, i.e. where the edges have binary states and the protocol can modify them, and consider the case in which the initial topology is always connected.

\begin{proposition} [NETs Impossibility of Termination] \label{pro:termination-impossibility}
There is no protocol that can compute with termination the predicate $(N_a\geq 1)$, even if the initial topology is connected and even if there is a pre-elected unique leader.
\end{proposition}
\begin{proof} 
Assume there is such a terminating protocol $A$. Take the initial topology $G=(V_I,E)$, which is equal to the clique $K_n$, and all nodes have input $b$. Any fair execution $s$ of $A$ on $G$ must reach a configuration in which all nodes have terminated saying ``no''. Next, consider the initial topology $G^\prime=(V_I\cup \{u\},E^\prime)$, which is the same as $G$ with an additional node $u$ that is connected to an arbitrary node of $V_I$. All nodes of $V_I$ have again input $b$ while $u$ has input $a$. Consider any fair execution $s^\prime$ on $G^\prime$ that has $s$ as an unfair prefix. Clearly, in $s^\prime$ all nodes in $V_I$ terminate during $s$ (as before) erroneously saying ``no'' without having interacted with $u$ that has input $a$. The crucial observation that we used to arrive at the result is that the initial states of the nodes are independent of the size of the network, so, even though $|V(G^\prime)|>|V(G)|$, we can assume that in both $G$ and $G^\prime$, input $x$ is translated to the same initial state $I(x)$ (the result would not hold in general if the initial state had a means of knowing $n$ in advance, as, e.g., a node could wait until it has marked $n$ distinct nodes, which would allow the protocol differentiate its behavior between $G$ and $G^\prime$). \footnote{Observe that the impossibility would not hold in some special cases of initial topologies, such as the spanning line, provided that the protocol knows this and also that the initial states of the nodes contain some information about the structure (e.g. the endpoints of the line know that they are the endpoints).} 
\end{proof} 

So, connectivity of the initial topology alone, even if the protocol is allowed to transform the topology, is not sufficient for non-trivial terminating computations. In the rest of the paper we shall naturally try to overcome this by adding to the model minimal and realistic extra assumptions. Interestingly, it will turn out that there are some very plausible such assumptions that allow for: (i) termination and (ii) computation of all predicates on input assignments that can be computed by a TM in quadratic space ($O(n^2)$, where $n$ is the number of nodes). 

One of the assumptions that we will keep throughout is that the nodes are capable of detecting some small local degrees. For example, in Section \ref{sec:unique-leader} we will assume that a node can detect that it has local degree 1 or 2, otherwise it knows that it has degree in $\{0,3,4,...,n-1\}$ without being able to tell its precise value. We will complement this local degree detection mechanism with either a unique leader or a common neighbor detection mechanism in order to arrive at the above strong characterization. 

Keep in mind that we want to give protocols for Acyclicity and Terminating Line Transformation. In Acyclicity, the protocol begins from any connected active topology and has to transform it to an acyclic network without ever breaking the connectivity, while in Terminating Line Transformation the protocol does not necessarily have to preserve connectivity but it has to satisfy the additional requirements its constructed network to be a spanning line and to always terminate. Still, even for the Terminating Line Transformation problem we shall mostly focus on protocols that perform the transformation without ever breaking connectivity. A justification of this choice, is that arbitrary connectivity breaking could render the protocol unable to terminate even if the protocol is equipped with all the additional mechanisms mentioned above. This is made formal in Lemma \ref{lem:line-impossibility} of Section \ref{subsec:impossibilities}. One way to appreciate this is to consider a protocol in which a leader breaks in some execution the network into an unbounded number of components. Then the leader can no longer distinguish an execution in which one of these components is being concealed from an execution that it is not. For example, if the leader is trying to construct a spanning line, then it has no means of distinguishing a spanning line on all nodes but the concealed ones from one on all nodes. Of course, this does not exclude protocols that perform some controlled connectivity breakings, e.g. a leader breaking a spanning line at one point and then waiting to reconnect the two parts. So, in principle, our problems could have been defined independently of whether connectivity is preserved or not, as they can also be solved in some cases by protocols that do not always preserve connectivity. However, in this work, for simplicity and clarity of presentation, we have chosen to focus only on those protocols that always preserve connectivity.

Before starting to present our protocols for above problems and the upper bounds on time provided by them, we give a lower bound on the time that any protocol needs in order to solve the Line Transformation problem.

\begin{lemma} [Line Transformation Lower Bound] \label{lem:line-lower-bound}
The running time of any protocol that solves the Line Transformation problem is $\Omega(n^2\log n)$.
\end{lemma}
\begin{proof} 
Consider the worst case of an active clique initially, in which case any protocol must deactivate a total of $\Theta(n^2)$ edges. The time needed to deactivate these edges is lower bounded by the time needed for the scheduler to pick each of these edges at least once. This is asymptotically equivalent to the time needed for the scheduler to perform an edge cover, which was analyzed in \cite{MS14} to require expected time $\Theta(n^2\log n)$. So, any protocol in this worst-case instance must perform an average of $\Omega(n^2\log n)$ steps.  
\end{proof}

\section{Transformers with a Unique Leader}
\label{sec:unique-leader}

We begin from the simplest case in which there is initially a pre-elected unique leader that handles the transformation. Recall that the initial active topology is connected. The goal is for the protocol to transform the active topology to a spanning line and when this occurs to detect it and terminate (i.e. solve the Terminating Line Transformation problem). Ideally, the transformation should preserve connectivity of the active topology during its course (or break connectivity in a controlled way, because, as we already discussed in the previous section, uncontrolled/arbitrary connectivity breaking may render termination impossible). Moreover, as a minimal additional assumption to make the problem solvable (in order to circumvent the impossibility of Proposition \ref{pro:termination-impossibility}), we assume that a node can detect whether it has local degree 1 or 2 (otherwise it knows that it has degree in $\{0,3,4,...,n-1\}$ without being able to tell its precise value). We first (in Section \ref{subsec:cycle-elimination}) give a straightforward solution, with a complete presentation of its transitions and an illustration showing them in action. Though that protocol is correct, it is rather slow and it mainly serves as a demonstration of the model and the problem under consideration. Then (in Section \ref{subsec:line-around-a-star}) we follow a different approach and arrive at a time-optimal protocol for the problem.

\subsection{A First Solution}
\label{subsec:cycle-elimination}

The idea is simple. The leader begins from its initial node and starts forming an arbitrary line by expanding one endpoint of the line towards unvisited nodes. Every such expansion either occurs over an edge that was already active from the very beginning or over an inactive edge which the protocol activates. Apart from expanding its active line with the goal of making it spanning after $n-1$ expansions, the leader must also guarantee that eventually no cycles will have remained. One idea would be to first form a spanning line and then start eliminating all unnecessary cycles, however there is, in general, no way for the protocol to detect that the line is indeed spanning, due to the possible presence of non-line active edges joining nodes of the line. This is resolved by eliminating line-internal cycles ``online'' after every expansion of the line. This guarantees that when the last expansion occurs and the protocol deactivates the last cycles, the active topology will be a spanning line. Now the protocol can easily detect this by traversing the line from left to right and comparing the observed active degree sequence to the target degree sequence $1,2,2,\ldots,2,1$ (i.e. the degree sequence of a spanning line). We next give the detailed description of the protocol.\\  

\noindent\textbf{Protocol Online-Cycle-Elimination.} The leader marks its initial node as ``left endpoint'' $e_l$ and picks an arbitrary next node for the line (for the first step it could be from its active neighbors, because there is at least one such node due to initial connectivity) and marks that node as ``right endpoint'' $e_r$. Then the leader moves to $e_r$, finds an arbitrary next node which is not part of the current line, if the edge is inactive it activates it and marks that node as $e_r$ and the previous $e_r$ is converted to $i$ (for ``internal node'' of the line). Observe that the active line is always in special states, which makes its nodes detectable.

After every such expansion, the leader starts a cycle elimination phase. In particular, the leader deactivates all edges that introduce a cycle inside its active line. To do this, it suffices after every expansion to deactivate the cycles introduced by the new right endpoint $e_r$. First, the leader moves to $e_l$ (e.g. by direct communication or by traversing the active line to the left). Every time, the leader waits to meet $e_r$, in order to check the status of the edge; if it is active, it deactivates it and then moves on step to the right on the line. When the leader arrives at the left neighbor of $e_r$, all line-internal cycles have been eliminated and the leader just moves to $e_r$. If the degree sequence observed during the traversal to the right (the degree of a node is checked \emph{after} checking and possibly modifying the status of its edge to $e_r$) was of the form $1,2,2,\ldots,2,1$ then the line is spanning and the leader terminates. Otherwise, the line is not spanning yet and the leader proceeds to the next expansion.

The code of the protocol is presented in Protocol \ref{prot:online-cycle}. For readability, we only present the code for the expansion and cycle elimination phases and we have excluded the termination detection subroutine (it is straightforward to extend the code to also take this into account). An illustration showing what are the roles of the various states and transitions during the expansion and cycle elimination phases, is given in Figure \ref{fig:cycle-elimination}.

\floatname{algorithm}{Protocol}
\renewcommand{\algorithmiccomment}[1]{// #1}
\begin{algorithm}[!h]
  \caption{Online-Cycle-Elimination}\label{prot:online-cycle}
  \begin{algorithmic}
    \medskip
    \State $Q=\{l_0,l_1,l_c,l_c^\prime,l,q_0,e,i,i^\prime,t,t_e,t_f,t_f^\prime,t_r,t^\prime,p,p^\prime,p^{\dprime}\}$, initially the unique leader is in state $l_0$ and all other nodes are in state $q_0$
    \State $\delta$: 
    \begin{align*}
    (l_0,q_0,1) &\ra (e,l_1,1) &
    (t_e,i,1) &\ra (e,t,1)&
    (p^{\dprime},t^\prime,1) &\ra (p,t,1)\\
    (l_1,q_0,\cdot) &\ra (i^\prime,l_c^\prime,1) &
    (t,l_c,\cdot) &\ra (t_r,l_c,0) &
    (t_r,i^\prime,1) &\ra (p^\prime,t_f^\prime,1)\\
    (e,l_c^\prime,\cdot) &\ra (t_e,l_c,0) &
    (t_r,i,1) &\ra (p^\prime,t^\prime,1) &
    (p^{\dprime},t_f^\prime,1) &\ra (i,t_f,1)\\
    (t_e,i^\prime,1) &\ra (e,t_f,1) &
    (e,p^\prime,1) &\ra (e,p^{\dprime},1) &
    (l,q_0,\cdot) &\ra (i^\prime,l_c^\prime,1)\\
    (t_f,l_c,1) &\ra (i,l,1) &
    (p,p^\prime,1) &\ra (i,p^{\dprime},1)
    \end{align*}
    \State \Comment {All transitions that do not appear have no effect}
    \State \Comment {The logical structure is better followed if the transitions are read from top to bottom}    
  \end{algorithmic}
\end{algorithm}

\begin{figure}[!hbtp]
\centering{
\includegraphics[width=0.65\textwidth]{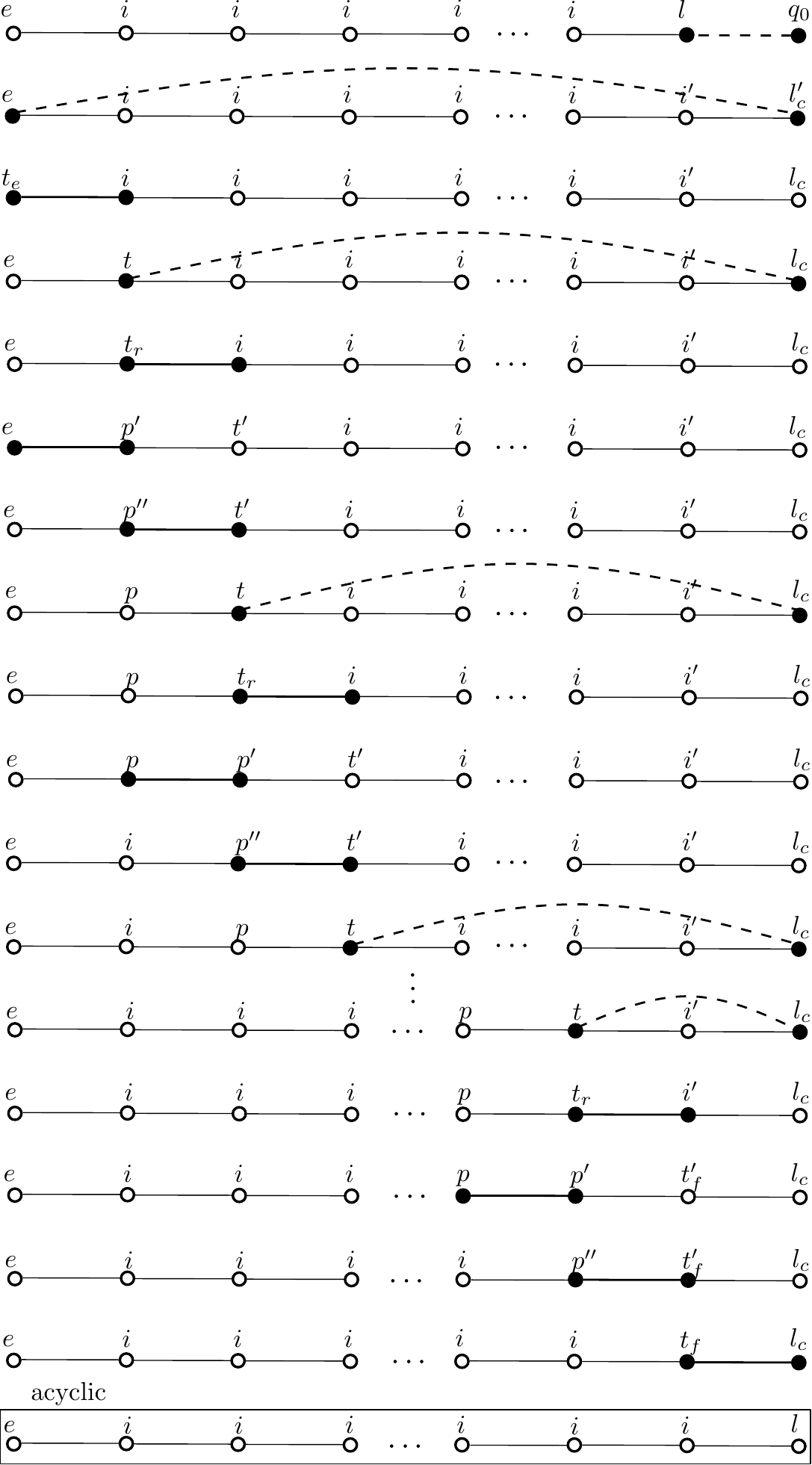}
}
\caption{An illustration of all transitions involved in Protocol Online-Cycle-Elimination during expansion of the line and elimination of the newly introduced internal cycles. The line at the top shows an expansion of the current acyclic line, the intermediate steps show the the process of eliminating cycles, and the line at the bottom is the new acyclic line. In every step, the two interacting nodes are colored black and joined by a bold edge. Dashed edges could be either active or inactive. The dashed edge of an expansion (top line) is activated no matter what its previous state was, while all other dashed edges in the figure, that correspond to (potential) cycle eliminations, are deactivated no matter what their state was.} \label{fig:cycle-elimination}
\end{figure}

\begin{theorem}
By assuming a pre-elected unique leader and the ability to detect local degrees 1 and 2, Protocol Online-Cycle-Elimination solves the Terminating Line Transformation problem in $\Theta(n^4)$ time.
\end{theorem}
\begin{proof}
We prove the following invariant: ``For all $1\leq i\leq n-1$, after the $i$th expansion and cycle elimination phases, the leader lies on the $e_r$ endpoint of an active line of (edge-)length $i$ without line-internal cycles (still any node of the line may have active edges to the rest of the graph) and the active topology is connected''. This implies that for $i<n-1$ there is at least one node of the line that has an edge to a node not belonging to the line and that for $i=n-1$ the active topology is a spanning line (without any other active edges). 

First observe that connectivity never breaks, because whenever the protocol deactivates an edge $e=uv$, both $u$ and $v$ are nodes belonging to the active line formed so far (in particular, at least one of them is the $e_r$ endpoint of the line). As $e$ is an edge forming a cycle on the active line after its deactivation connectivity between $u$ and $v$ still exists by traversing the line.

We prove by induction the rest of the invariant. It holds trivially for $i=1$. Given that it holds for any $1\leq i\leq n-2$ we prove that it holds for $i+1$. By hypothesis, when expansion $i+1$ occurs, the only possible line-internal cycles are between the new $e_r$ and the rest of the line. During the cycle elimination phase the protocol eliminates all these cycles, and as a result by the end of phase $i+1$ the active line has now length $i+1$, it has no internal cycles and is still connected to the rest of the graph. 

It remains to show that the leader terminates just after phase $n-1$ and never at a phase $i<n-1$. For the first part, after phase $n-1$ the active topology is a spanning line, thus the observed degree sequence when the leader traverses it from left to right is of the form $1,2,2,\ldots,2,1$ which triggers termination. For the second part, after any phase $i<n-1$ there is at least one node of the line having an active edge leading outside the line. In case that node is an endpoint, its active degree is at least 2 and in case it is an internal node its active degree is at least 3 (after eliminating a possible cycle of that node with $e_r$). So, in this case the observed degree sequence is not of the form $1,2,2,\ldots,2,1$ and, as required, the leader does not terminate.

For the running time, the worst case is when the initial active topology is a clique. In this case, the protocol must deactivate $\Theta(n^2)$ edges to transform the clique to a line. Every edge deactivation is performed by placing a mark on each endpoint of the edge and waiting for the scheduler to pick that edge for interaction. This takes time $\Theta(n^2)$, so the total time for deactivating $\Theta(n^2)$ edges is $\Theta(n^4)$.
\end{proof}

We should mention that due to the unique-leader guarantee, it suffices to only have detection of whether the degree is equal to 1 (i.e. the detection of degree equal to 2 can be dropped). The reason is that the leader can every time break the line at some point while marking the two endpoints of the edge and then check whether one of these nodes has degree 1. If yes, then its previous degree was 2 and the leader waits for the two marked nodes to interact again in order to reconnect them, now knowing their degree.

A drawback of the above protocol is that it is rather slow. In the sequel, we develop another protocol, based on a different transformation technique, which is time-optimal. 

\subsection{An Optimal Protocol}
\label{subsec:line-around-a-star}

\noindent\textbf{Protocol Line-Around-a-Star.} There is initially a unique leader in state $l$ and all other nodes are in state $q_0$. Moreover, nodes can detect when their degree is 1.

The leader starts connecting with the $q_0$s (by activating the connection between them in case it was inactive and by preserving it in case it was already active) and converts them to $p^\prime$ trying to form a star with itself at the center. When two $p^\prime$s interact, if the edge is active they deactivate it, trying to become the peripherals of the star. Additionally, if after such a deactivation the degree of a $p^\prime$ is 1, then the $p^\prime$ becomes $p$ to represent the fact that it is now connected only to the leader and has become a normal peripheral. The same occurs if after the interaction of the leader with a $q_0$, the degree of the $q_0$ is 1, i.e. the $q_0$ immediately becomes a normal peripheral $p$.

When the leader first encounters a $p$, it starts constructing a line which has as its ``left'' endpoint the center of the star and that will start expanding over the peripherals until it covers them all. Whenever the leader interacts with an internal node of the line, it disconnects from it (but it never disconnects from second node of the line, counting from the center; to ensure this, the protocol has that node in a distinguished state $i^\prime$ while all other internal nodes of the line are in state $i$). The protocol terminates when the degree of the center becomes 1 for the first time (note that it could be 1 also at the very beginning of the protocol but this early termination can be trivially avoided). An example execution is depicted in Figure \ref{fig:around-a-star}.

\begin{figure}[!hbtp]
   \centering{
        \subfigure[]{
        \includegraphics[width=0.4\textwidth]{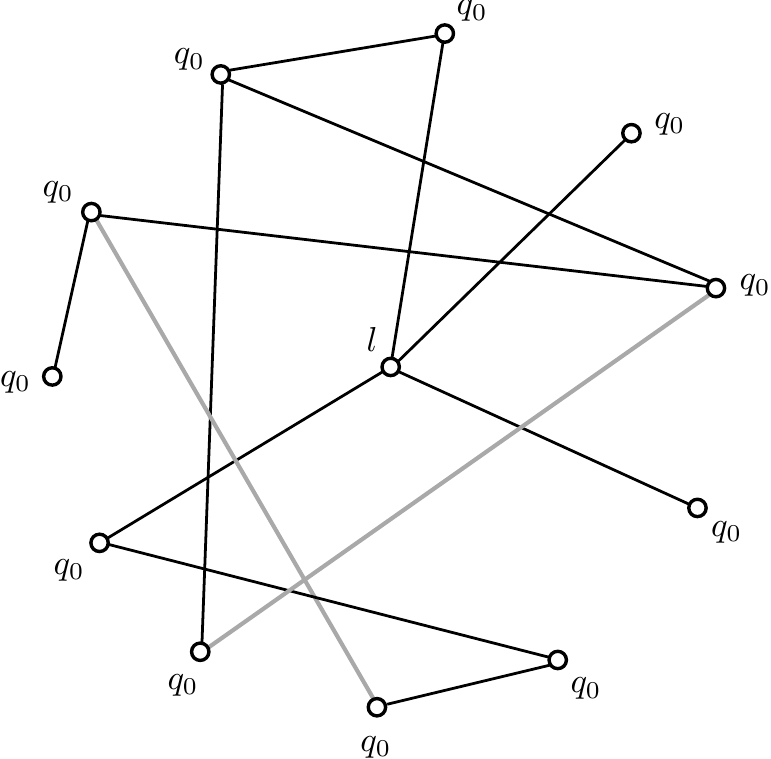}
        \label{fig:around-a-star1}}
	\hspace{1cm}
        \subfigure[]{
        \includegraphics[width=0.4\textwidth]{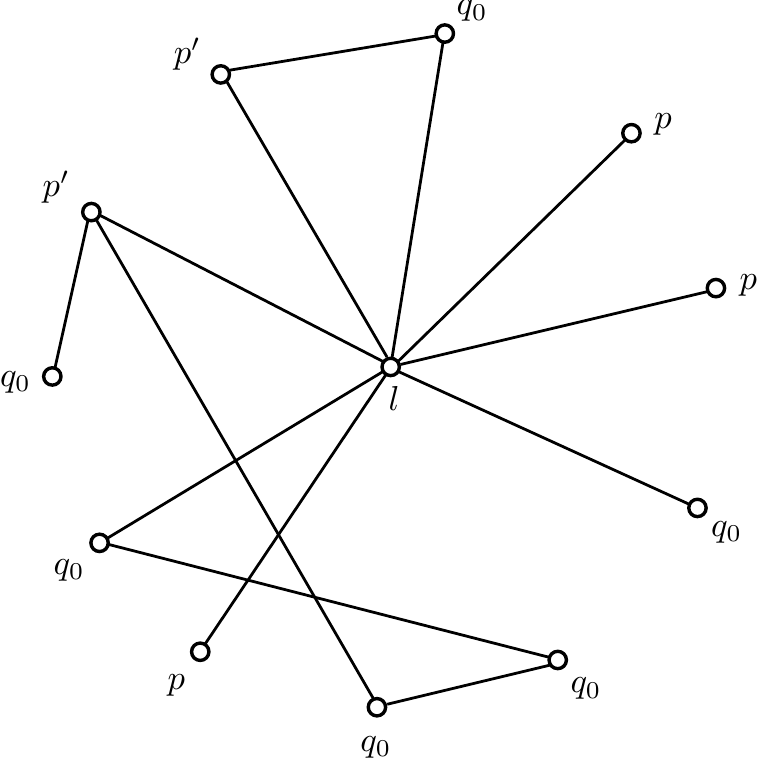}
        \label{fig:around-a-star2}}
	\hspace{1cm}
        \subfigure[]{
        \includegraphics[width=0.4\textwidth]{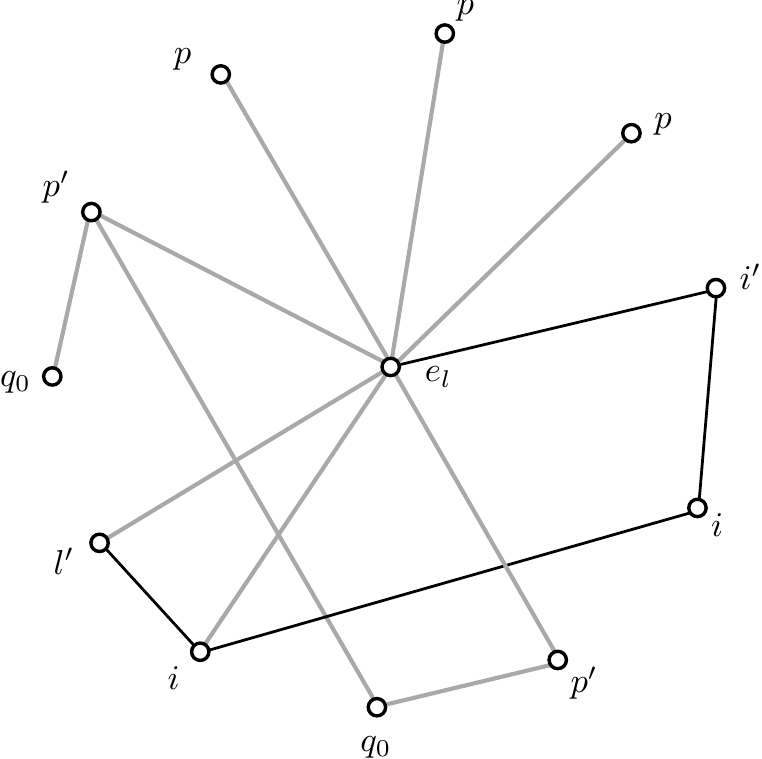}
        \label{fig:around-a-star3}}
	\hspace{1cm}
        \subfigure[]{
        \includegraphics[width=0.4\textwidth]{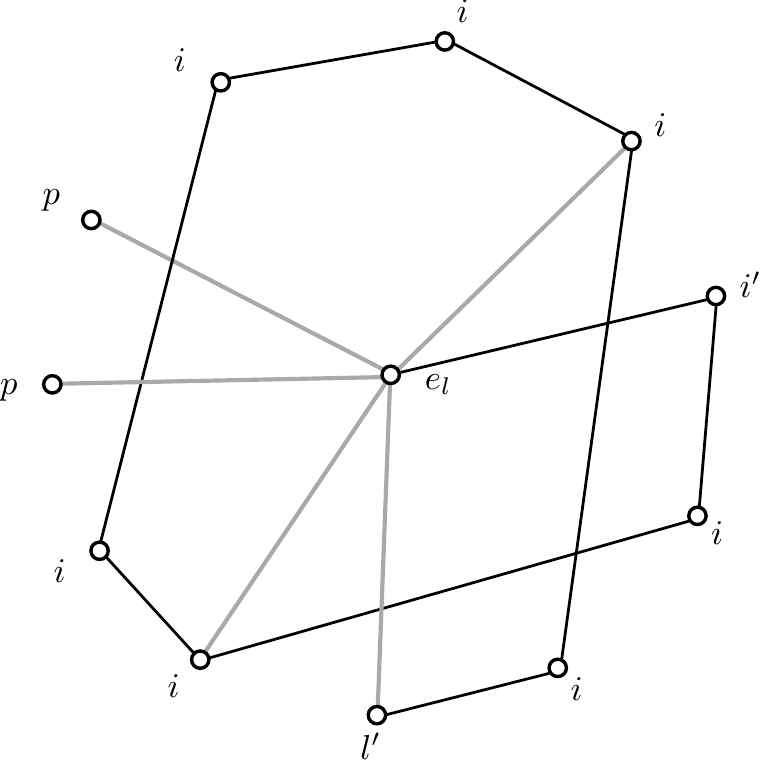}
        \label{fig:around-a-star4}}        
        }
   \caption{An example execution of Protocol Line-Around-a-Star. In all subfigures, black and gray edges are active and missing edges are inactive. Black and gray are used together whenever we want to highlight some subnetwork of the active network. (a) Initially there is a unique leader in state $l$, all other nodes are in state $q_0$, and there is a connected active topology (here, the black edges just highlight a connected spanning subgraph of the initial topology). (b) After a a few steps, some $q_0$s have been converted to $p^\prime$ by $l$ and some of them have already lost some edges to other peripherals. Some other $q_0$s have been converted to normal peripherals in state $p$; these only have a single neighbor, the leader. (c) More peripherals have been created. Additionally, the leader has started to form its line (drawn by black edges) over the peripherals. The center is in a new leader state $e_l$ and the other endpoint of the line is in state $l^\prime$. The second node of the line is in state $i^\prime$ (to avoid its disconnection from $e_l$) while all other internal nodes are in state $i$. The center has already disconnected from some of them. (d) The line is almost spanning and the center has disconnected from most internal nodes of the line.} \label{fig:around-a-star}
\end{figure}

\begin{theorem}
By assuming a pre-elected unique leader and the ability to detect local degree 1, Protocol Line-Around-a-Star solves the Terminating Line Transformation problem. Its running time is $\Theta(n^2 \log n)$, which is optimal.
\end{theorem}
\begin{proof}
We begin with correctness. First of all, observe that every $q_0$ eventually becomes $p$, i.e. a peripheral that at that point is only connected to the center. This follows because the center forever attracts the $q_0$s making them $p^\prime$ and a $p^\prime$ only disconnects from other peripherals until it becomes $p$. This implies that eventually each non-leader node will become available for the line to expand over it and thus the line will eventually become spanning. Next observe that the protocol never disconnects the topology. In particular, the protocol performs only two types of edge eliminations: (i) $(p^\prime,p^\prime)$ which cannot disconnect because the $p^\prime$s are still connected via the center and (ii) (center, node $3\leq i\leq k$ of the line of length $k$) which again cannot disconnect because all nodes of the line are still connected to each other via the line. It remains to show that the protocol terminates iff the active topology has become a spanning line. It suffices to show that after the line formation subroutine has performed at least on step, the degree of the center first becomes 1 when the active topology becomes a spanning line. Clearly, when the active topology becomes a spanning line the degree of the center must be 1 due to the fact that the center is one endpoint of the line. Now assume that the degree of the center becomes 1 while still the active topology is not a spanning line. As the line formation routine has already started, the unique edge of the center is an edge of the line covering nodes in $S\subsetneq V$. Moreover, all other nodes of the line were $p$s just before the line expanded on them, so their only possible edges are either edges of the line or edges to the center. It follows, that there can be no edge in the $(S,V\bs S)$-cut which contradicts the fact that the protocol never disconnects the topology. Hence, the degree of the center cannot have become 1 before the active topology is a spanning line and it follows that the protocol is correct.

We now analyze the running time of the protocol. First consider the time needed for the leader to connect to every $q_0$ (and convert all $q_0$ to $p^\prime$). This is equivalent to the time needed for a particular node to meet every other node, which is the meet everybody fundamental process analyzed in \cite{MS14} to take $\Theta(n^2 \log n)$ expected time. Next consider the time for all peripherals to disconnect from one another and become $p$. If we study this after the time all $q_0$ have become $p^\prime$, it is the time (in the worst case) needed for all edges to be picked by the scheduler. This is an edge cover which is known from \cite{MS14} to take expected time $\Theta(n^2 \log n)$. After the completion of both the above, we have a star with the leader at the center and all peripherals are only connected to the leader (and also possibly the formation of the line has already covered some of the peripherals). Next consider the formation of the line over the peripherals (as already states this may have begun in parallel to the above processes, however we may study the worst case in which it begins only after the above processes are complete). The right endpoint of the line is always ready for expansion towards another available peripheral. The time needed for the line to cover all peripherals is again the time of a meet everybody. It is the time needed for one particular state (which is the right endpoint of the line in this case) to meet every peripheral, every time being able to expand only towards peripherals that
are not yet part of the line. So, this part of the process takes time $O(n^2 \log n)$ to complete. We finally take into account the time needed for the center to disconnect from the peripherals that are part of the line. Again we can study this after the line has become spanning. This is simply a star deformation, i.e. the time needed until the center meets all peripherals in order to disconnect from them. This again is a meet everybody taking time $O(n^2 \log n)$. Taking all this into account, we conclude that the running time of the protocol is $O(n^2 \log n)$, which matches the $\Omega(n^2 \log n)$ lower bound of Lemma \ref{lem:line-lower-bound}, therefore the protocol is time-optimal.
\end{proof}

\section{Transformers with Initially Identical Nodes}
\label{sec:identical-nodes}

An immediate question, given the optimal Line-Around-a-Star protocol, is whether the unique leader assumption can be dropped and still have a correct and possibly also optimal protocol for Terminating Line Transformation. At a first sight it might seem plausible to expect that the problem is solvable. The reason is that the nodes can execute a leader election protocol (e.g. the standard pairwise elimination protocol; see e.g. \cite{AR09}) guaranteeing that eventually a single leader will remain in the system which can from that point on handle the execution of one of the leader-based protocols of the previous section. The only additional guarantee is to ensure that nothing can go wrong as long as there are more than one leaders in the population. Typically, this is achieved in the population protocol literature by the reinitialization technique in which the configuration of the system is reinitialized/restored every time another leader is eliminated so that when the last leader remains a final reinitialization gives a correct system configuration for the leader to work on. In fact, this technique and others have been used in the population protocol literature to show that most population protocol models do not benefit in terms of computational power from the existence of a unique leader (still they are known to benefit in terms of efficiency). 

In contrast to this intuition, we shall prove in this section (see Corollary \ref{cor:line-impossibility}) that if all nodes are initially identical, Terminating Line Transformation becomes impossible to solve (with the modeling assumptions we have made so far). In particular, we will show that any protocol that makes the active topology acyclic, may disconnect it in some executions in $\Theta(n)$ components (see Corollary \ref{cor:acyclicity-impossibility}). As already discussed in Section \ref{subsec:fi}, such a worst-case disconnection is severe for any terminating protocol, because, in this case, a component has no means of determining when it has interacted with (or heard from) all other components in the network.  

\begin{observation}
For a protocol to transform any topology to a line (or in general to an acyclic graph) without breaking connectivity, it must hold that the protocol deactivates an edge only if the edge is part of a cycle. Because deleting an edge $e$ of an undirected graph does
not disconnect the graph iff $e$ is part of a cycle.
\end{observation}

There are several ways to achieve this when there is a unique leader.
However, it will turn out that this is not the case when all nodes are initially
identical.

\subsection{Detecting Small Cycles Stably}

We begin by showing that if we just want to detect the existence of
small cycles stably (i.e. eventually all nodes say 1 if there is a cycle of a particular small length and all say 0
if there is no such cycle, without the need of termination) then there is a protocol that achieves this.

For simplicity, assume that we have a static directed connected active topology,
communication is clique (i.e. nodes also communicate via inactive edges),
all nodes are initially $(l,0)$ (i.e. leaders with output 0) and we want to
stably decide whether there is a directed 2-cycle. The code of the protocol is presented in Protocol \ref{prot:2cycle}. Note that in the transitions we do not specify a right-hand-side value for the edges, because the protocol never modifies their state.

\floatname{algorithm}{Protocol}
\renewcommand{\algorithmiccomment}[1]{// #1}
\begin{algorithm}[!h]
  \caption{\emph{Stable-2-Cycle-Detection}}\label{prot:2cycle}
  \begin{algorithmic}
    \medskip
    \State $Q=\{l,l^\prime,f,f^\prime\}\times\{0,1\}$, initially all nodes are in state $(l,0)$
    \State $\delta$: 
    \begin{align*}
(l,\cdot), (l,\cdot), \cdot &\ra (l,0), (f,0)   \mbox{ // leaders are pairwise eliminated and
restore their output to 0}\\
(l,0), (f,\cdot), 1 &\ra (f^\prime,0), (l^\prime,0) \mbox{ // marking in order to detect a 2-cycle}\\
(f,\cdot), (l,0), 1 &\ra (l,0), (f,0) \mbox{ // the leader moves nondeterministically to avoid being trapped}\\
(l^\prime,0), (f^\prime,0), 1 &\ra (l,1), (f,1) \mbox{ // a 2-cycle detected (possibly wrong
decision). The output becomes 1.}\\
(l,x), (f,\cdot), 0 &\ra (l,x), (f,x) \mbox{ // the follower copies the leader's
output}\\
(f,\cdot), (l,x), 0 &\ra (f,x), (l,x)\\
(l,1), (f,0), 1 &\ra (l,1), (f,1) \mbox{ // the follower copies the leader's
output}\\
(f,0), (l,1), 1 &\ra (f,1), (l,1)\\
(f^\prime,0), (l^\prime,0), 1 &\ra (f,0), (l,0) \mbox{ // unsuccessful attempt of detecting
2-cycle. Restoring.}\\
(l^\prime,0), (f^\prime,0), 0 &\ra (l,0), (f,0) \mbox{ // unsuccessful attempt of detecting
2-cycle. Restoring.}\\
(f^\prime,0), (l^\prime,0), 0 &\ra (f,0), (l,0)
    \phantom{\hspace{10cm}}
    \end{align*}
  \end{algorithmic}
\end{algorithm}

The protocol works as follows. If there is a unique leader $(l,0)$ and all other nodes are $f$, then the leader performs a nondeterministic search while trying to detect a 2-cycle. As there are no
conflicts with other leaders and wrong $f^\prime$, if there is a 2-cycle
eventually it will be detected, the leader will output 1 and all followers
will copy this output thus all eventually will output 1. If there is no
2-cycle, then the leader's output will forever remain 0 and all followers
will eventually copy this.

The only problem is while there are still more than one leaders. Pairwise
eliminations guarantee that eventually a single leader will remain.
However, when this occurs we want to be sure that there is no $f^\prime$ left by
previous rounds because this could possibly confuse the final leader and make it
detect a 2-cycle that does not exist. This is ensured by allowing
elimination only between two $l$s which implies that even if the eliminated
leader had introduced some $f^\prime$ in the past, before being eliminated it has
restored an $f^\prime$ (though not necessarily its own). 

We now turn this into a formal proof.

\begin{proposition}
Protocol Stable-2-Cycle-Detection stably decides whether the active topology contains a directed 2-cycle. In particular, when there is a 2-cycle all nodes stabilize to output 1 and when there is no 2-cycle all nodes stabilize to output 0.
\end{proposition}
\begin{proof}
We have to prove that eventually a unique $(l,0)$ leader remains and all other nodes are $(f,\cdot)$ followers. Because if this holds, then the protocol is correct as follows. If there is no 2-cycle, then the rule $(l^\prime,0), (f^\prime,0), 1 \ra (l,1), (f,1)$ is never triggered and as this is the only rule that can change the leader's output from 0 to 1, it follows that the leader's output remains forever 0.  Consider now the case where there is a 2-cycle between $u$ and $v$. Observe that as long as the leader has not detected a 2-cycle, it moves nondeterministically over the active edges. In particular, an outgoing edge is traversed by the rule $(l,0), (f,\cdot), 1 \ra (f^\prime,0), (l^\prime,0)$ followed by a restoring rule $(f^\prime,0), (l^\prime,0), 1 \ra (f,0), (l,0)$ and an incoming edge is followed by the rule $(f,\cdot), (l,0), 1 \ra (l,0), (f,0)$. Fairness implies that $u$ is eventually reached and the consecutive application of rules $(l,0), (f,\cdot), 1 \ra (f^\prime,0), (l^\prime,0)$ and $(l^\prime,0), (f^\prime,0), 1 \ra (l,1), (f,1)$ detects the 2-cycle and makes the leader's output 1. Again in this case, the leader's output forever remains 1 because only an interaction with another leader could change the leader's output from 1 to 0, but there is no other leader. In both cases the leader's output becomes correct and remains correct forever and eventually all followers will interact with the leader and copy the leader's output.

We now prove our initial claim. We first prove that an $l^\prime$ can always be restored to $l$. To this end observe that the number of $l^\prime$s is always equal to the number of $f^\prime$s, because $l^\prime$s and $f^\prime$s are always introduced and eliminated concurrently, as e.g. in the rules $(l,0), (f,\cdot), 1 \ra (f^\prime,0), (l^\prime,0)$ and $(f^\prime,0), (l^\prime,0), 1 \ra (f,0), (l,0)$. So, as long as there is an $l^\prime$, there is also an $f^\prime$ somewhere in the system, and any interaction between them will restore $l^\prime$ to $l$ (and also $f^\prime$ to $f$). This implies that as long as there are more than one leaders, a configuration in which at least two of them are in state $l$ is always reachable and due to fairness it is eventually reached. Then the two $l$ leaders may interact and one of them will become eliminated. Thus, pairwise elimination of leaders is never blocked and the elimination between the last two leaders will leave the system with a unique $(l,0)$ leader. At that point, all other nodes will be $f$ followers, because there are no $l^\prime$ leaders and, as proved above, the number of $f^\prime$ followers is always equal to the number of $l^\prime$ leaders, so there can be no $f^\prime$ followers. This completes the proof of our initial claim and correctness of the protocol follows.
\end{proof}

\begin{remark}
The above protocol can be easily extended, by using more marks, to decide the existence of undirected $k$-cycles, for any constant $k$.
\end{remark}

\subsection{Impossibility Results}
\label{subsec:impossibilities}

An immediate question is whether there is also a protocol with initially identical nodes that decides the existence of small cycles and additionally always terminates. We shall now show that this is not the case. In fact, this is already true for the Stable-2-Cycle-Detection protocol, because as long as there are more than one leaders, cycle detection is never certain; a detection decision may always be wrong and this is why the above protocol is always ready to backtrack it. The following first impossibility states that if all nodes are initially identical, any terminating decision concerning the detection of a small cycle may lead to a wrong decision in some other execution. We give the lemma in terms of decisions concerning the deactivation of an edge of a 3-cycle, however the same proof technique can be extended both to $k$-cycles, for any constant $k$, and to terminating decisions of existence (instead of deactivations). A variation of Lemma \ref{lem:first-impossibility} concerning the terminating decision of existence of directed 2-cycles was proved in \cite{CMNS13}. We here give a version concerning edge-deactivations that better fits our perspective and then generalize this in Theorem \ref{the:strong-impossibility} to obtain a much stronger impossibility result.

\begin{lemma} [A First Impossibility] \label{lem:first-impossibility}
If all nodes are initially identical, then there is no protocol that never disconnects the active topology and on every $G$ that contains a 3-cycle the protocol deactivates an edge of a 3-cycle of $G$.
\end{lemma}
\begin{proof}
Let $A$ be a protocol that from any initial connected active topology in which all nodes are in state $q_0$, stabilizes to
a topology without 3-cycles still without ever disconnecting the network.

Take a triangle $T$: $u_1,u_2,u_3$. For every fair execution $s$, of $A$ on $T$, we have that in a finite number of steps $A$ deactivates some
edge of the triangle, because it has to break the 3-cycle in a finite number of steps. 

Now take a hexagon $H$ (6-cycle): $v_1,v_2,v_3,v_1^\prime,v_2^\prime,v_3^\prime$. Consider now any fair execution $s^{\dprime}$ of $A$ on $H$ which has as a prefix the following unfair execution $s^\prime$, that we construct by mimicking the behavior of the scheduler in execution $s$: Every time $(u_i,u_{i+1})$, for $i\in\{1,2\}$, is selected in $s$ (i.e. on the triangle $T$), both $(v_i,v_{i+1})$ and $(v_i^\prime,v_{i+1}^\prime)$ are selected in $s^\prime$ (i.e. on the hexagon $H$) and every time $(u_3,u_1)$ is selected in $T$, both $(v_3,v_1^\prime)$ and $(v_3^\prime,v_1$ are selected in $s^\prime$ (the way the interactions are matched between the two executions, is depicted in Figure \ref{fig:3-cycle}). Observe that these are the only possible types of interactions between nodes of $H$ that we allow during the unfair execution $s^\prime$ and we are allowed to add any finite unfair prefix to a fair execution. It is not hard to see that during the unfair schedule, $v_i$ and $v_i^\prime$ always have the same state as $u_i$ (and with each other). So, when an interaction $(u_i,u_{i+1})$ or $(u_3,u_1)$ first deactivates an edge in $s$, both the interactions $(v_i,v_{i+1})$ and $(v_i^\prime,v_{i+1}^\prime)$ or $(v_3,v_1^\prime)$ and $(v_3^\prime,v_1$ are selected in $s^\prime$ and both deactivate the corresponding edges of $H$. As two opposite edges of the hexagon are deactivated, the hexagon becomes disconnected. In fact, the symmetry can have any desired length by $n/3$ repetitions of the pattern. Therefore, we conclude that there can be no protocol $A$ as assumed above.

In summary, we have shown that if a protocol eventually deactivates an edge of a 3-cycle then there must be some execution of the protocol on some other topology that leads to disconnection.
\end{proof}

\begin{figure}[!hbtp]
\centering{
\includegraphics[width=0.75\textwidth]{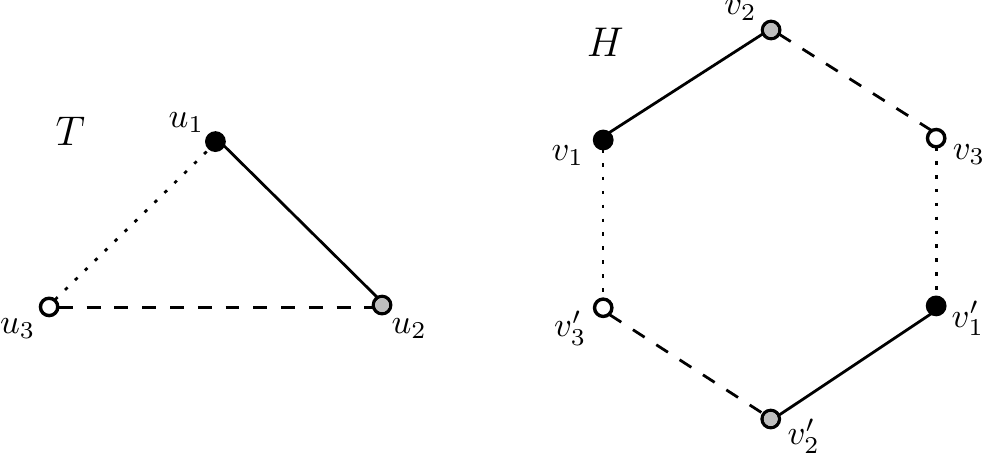}
}
\caption{The triangle $T$ and the hexagon $H$ used in the proof of Lemma \ref{lem:first-impossibility}. The nodes and edges of the two graphs have been marked according to the equivalencies between executions $s$ and $s^\prime$ mentioned in the proof. For example, the state of the gray node $u_2$ just after step $i$ of $s$ will be the same as the states of the gray nodes $v_2$ and $v_2^\prime$ just after step $2i$ of $s^\prime$.} \label{fig:3-cycle}
\end{figure}

Lemma \ref{lem:first-impossibility} implies that, if all nodes are initially identical, there is no protocol for the Acyclicity problem (requiring the processes to stably make the input graph acyclic without ever disconnecting it). In fact, the above is a general impossibility in case the class of input graphs under consideration contains also the small triangle (of $n=3$). On this triangle, it happens that the protocol cannot activate any edges before deactivating at least one which leads in disconnecting the large cycle (e.g. the hexagon). The above, though, does not imply an impossibility for the subclass of graphs that have at least one inactive edge and on which the protocol may first activate some edges before deactivating.

We now obtain a much stronger impossibility that resolves this.

\begin{theorem} [Strong Impossibility] \label{the:strong-impossibility}
For every connected graph $G$ with at least one cycle, there is an infinite family of graphs $\cal{G}$ such that for every $G^\prime\in \cal{G}$ every protocol (beginning from identical states on all nodes) that makes $G$ acyclic may disconnect $G^\prime$ in some executions.
\end{theorem}
\begin{proof}
Let $G$ consist of $n$ nodes $u_1,u_2,\ldots,u_n$. Let $u_iu_j$ be an edge
belonging to a cycle of $G$ (at least one such edge exists). Create $k\geq 2$
copies of the nodes of $G$: $\{v_{11},v_{12},\ldots,v_{1n}\}$, $\{v_{21},v_{22},\ldots,v_{2n}\},\ldots,$
$\{v_{k1},v_{k2},\ldots,v_{kn}\}$. The parameter $k$ can be as large as we want, this is why
the family $\cal{G}$ is infinite. Make every copy almost identical to $G$ by
introducing all edges that $G$ has, apart from the edge $u_iu_j$, i.e. in
every copy $\{v_{h1},v_{h2},\ldots,v_{hn}\}$ we introduce $\{v_{ht},v_{hz}\}$ iff
$\{u_t,u_z\}\in E(G)$ and  $\{u_t,u_z\}\neq \{u_i,u_j\}$. In words, we have $k$ copies
of $G$ all missing the edge corresponding to $u_iu_j$. Observe that each of
these copies is a connected component because G was connected after the
removal of $u_iu_j$. Finally we introduce a large ``cycle'' connecting all
these components. The cycle uses the nodes corresponding to $u_i$ and $u_j$
interchangeably. In particular, the cycle is formed by adding the edges
$\{v_{1i},v_{2j}\},\{v_{2i},v_{3j}\},\ldots,\{v_{(k-1)i},v_{kj}\},\{v_{ki},v_{1j}\}$. This completes the
construction of $G^\prime$ (see also Figure \ref{fig:impossibility} for a simple example).

\begin{figure}[!hbtp]
   \centering{
        \subfigure[]{
        \includegraphics[width=0.18\textwidth]{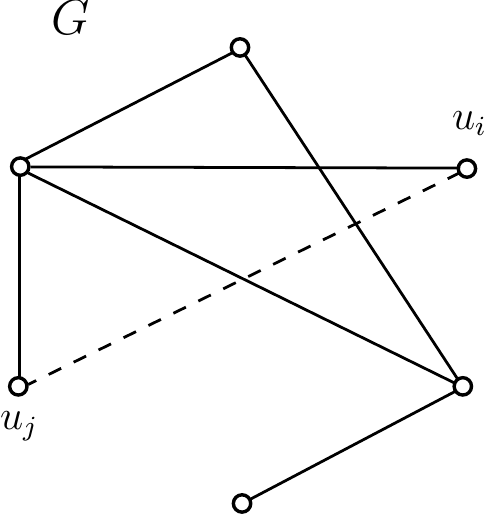}
        \label{fig:impossibility1}}
	\hspace{1cm}
        \subfigure[]{
        \includegraphics[width=0.58\textwidth]{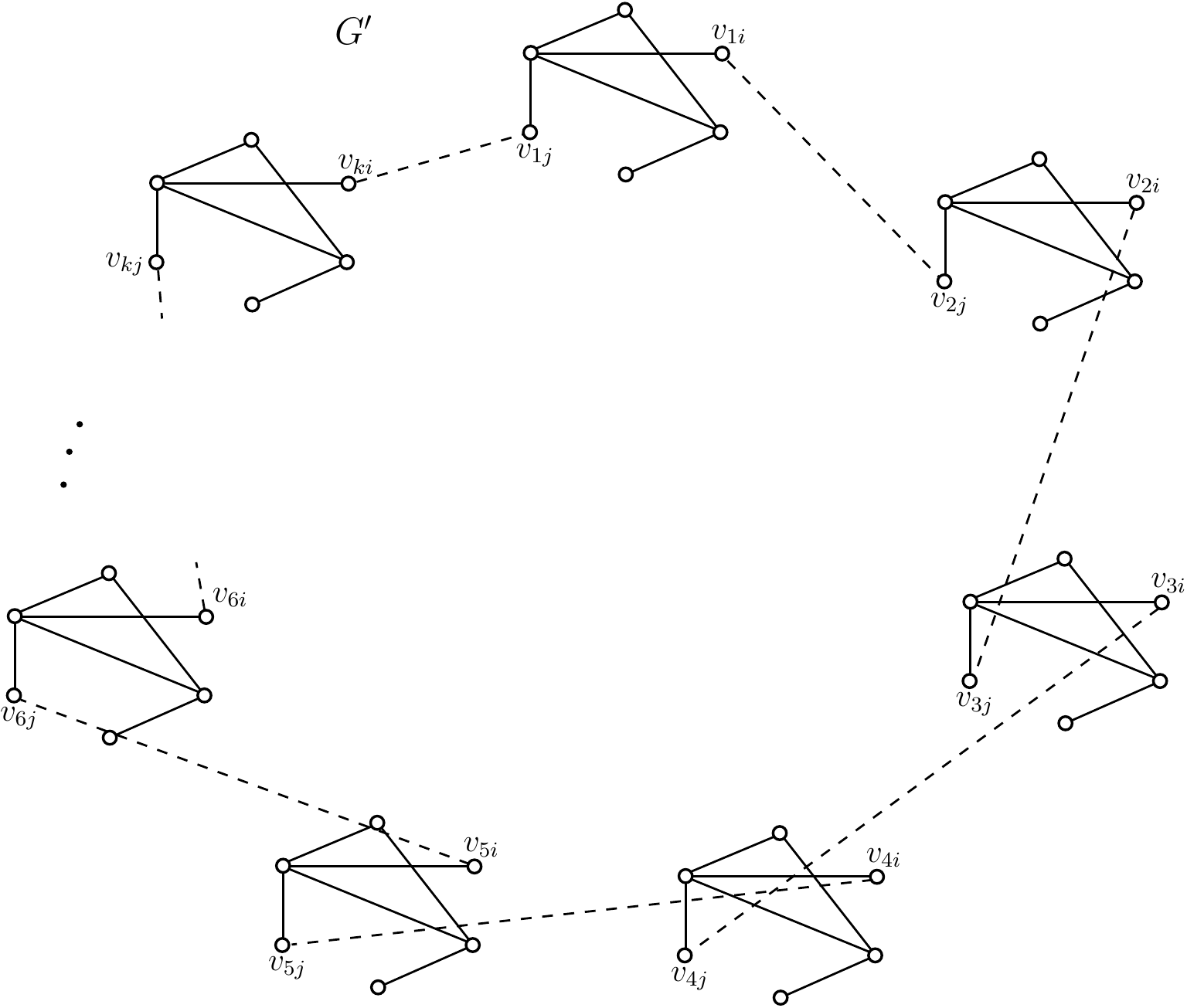}
        \label{fig:impossibility2}}
        }
   \caption{(a) The graph $G$. The selected edge $u_iu_j$, which is part of a cycle, is drawn dashed. (b) The graph $G^\prime$ (in fact the whole subfamily of $\cal{G}$ for these particular $G$ and $u_iu_j$; only the number of components changes between $G^\prime$s). Each component is identical to $G$ apart from edge $u_iu_j$ which has been removed and has been replaced by the intercomponent edges $\{v_{zi},v_{(z+1)j}\}$, $1\leq z\leq k$ and $k+1\coleq 1$.} \label{fig:impossibility}
\end{figure}

Now $G^\prime$ is clearly connected (and in fact also contains at least one large
cycle containing at least nodes $v_i$ and $v_j$ from each component) because if
$x,w$ are two nodes of $G^\prime$ then (i) if $x$ and $w$ are in the same component there is
a path between them because each component is connected and (ii) if $x$ and
$w$ belong to different components, there is a path from $x$ to $v_i$ in the
component of $x$, then an edge to $v_j$ of the next component and a path to the
$v_i$ of that component, and by following all these $v_i,v_j$ connections we can
finally arrive at the $v_j$ of the component of $w$ and reach $w$.

Now let $A$ be any protocol that when executed on any initially connected
graph $H$ with all nodes initially identical, eventually makes $H$ acyclic
without ever breaking the connectivity of $H$.

Such a protocol $A$ must work correctly on both $G$ and $G^\prime$. So, in every fair
execution $\alpha$ on $G$, $A$ eventually makes $G$ acyclic (by possibly activating and
deactivating edges in the meantime; i.e. we do not pose any further
restriction on $A$'s behaviour). But $A$ must also work correctly on $G^\prime$.
Consider now the following finite unfair execution $\alpha^\prime$ of $A$ on $G^\prime$ (we are
allowed to use any finite unfair prefix of a fair execution): $\alpha^\prime$ mimics $\alpha$.
In particular, whenever $\alpha$ selects edge $\{u_t,u_z\}$ (either active or
inactive) of $G$ different than $\{u_i,u_j\}$ (i.e. the one that was excluded from
the components), $\alpha^\prime$ selects one after the other all copies of $\{u_t,u_z\}$
inside the components of $G^\prime$, i.e. first its copy in the first component,
then its copy in the second component, and so on. If the edge is
$\{u_i,u_j\}$, then $\alpha^\prime$ again selects one after the other all its copies in $G^\prime$,
that is the edges joining the components. This completes the description
of the unfair finite prefix of $\alpha^\prime$.

Now it only remains to give a proof of local indistinguishability of all
nodes between $\alpha$ and $\alpha^\prime$. This can be proved by induction. We show that
whenever $\alpha$ selects an edge and the two nodes interact and update their
states and the state of the edge, then after the above described
selections of $\alpha^\prime$, all copies of the nodes in $G^\prime$ have the same states as
they do in $G$ and also their neighborhood structure including edges and
states is the same as it is in $G$. In fact, if the edge is different than
$\{u_i,u_j\}$ then this holds trivially because the same update that we have in
$G$ is performed in all components of $G^\prime$. If the edge is $\{u_i,u_j\}$ then a
deactivation cannot occur because this would deactivate all of its copies
between components in $G^\prime$ and would disconnect $G^\prime$ (contradicting
correctness of $A$). So only state updates can be performed in this case and
symmetry between the two executions is preserved.

Now observe the following. $\alpha$ eventually makes $G$ acyclic and still
connected (i.e. a tree). This means that $G$ has eventually precisely $n-1$
active edges. Every such edge has precisely $k$ copies in $G^\prime$. Due to the
above indistinguishability argument, the only edges of $G^\prime$ that are active
in $\alpha^\prime$ (after all the corresponding simulations of $\alpha$) are the copies of
these $n-1$ edges, that is $\alpha^\prime$ leaves in a finite number of steps $G^\prime$ with
$k(n-1)=kn-k$ edges. But $G^\prime$ consists of $kn$ nodes, thus it needs at least
$kn-1$ edges to be connected and $kn-k < kn-1$ for all $k\geq 2$. It follows that
$\alpha^\prime$ disconnects $G^\prime$ which contradicts the existence of such a correct
protocol $A$.
\end{proof}

The above stronger result states that \emph{every} connected graph $G$ has a corresponding infinite family of graphs (in most cases disjoint to
the families of other graphs) such that Acyclicity cannot be solved at the same time on $G$ and on a $G^\prime$ from the family. This means that it does not just happen for Acyclicity to be unsolvable in a few specific inconvenient graphs. \emph{All graphs} are in some sense 
inconvenient for Acyclicity when studied together with the families that we have defined.

\begin{corollary} [Acyclicity Impossibility] \label{cor:acyclicity-impossibility}
If all nodes are initially identical, then any protocol that always makes the active topology acyclic may disconnect it in some executions in $\Theta(n)$ active components (i.e. in a worst-case manner).
\end{corollary}

\begin{lemma} \label{lem:line-impossibility}
If a protocol breaks in some executions the active topology into $\Theta(n)$ components, then such a protocol cannot solve the Terminating Line Transformation problem.
\end{lemma}
\begin{proof}
Let $A$ be a protocol that solves Terminating Line Transformation. Take any $G$ and its corresponding infinite family $\cal{G}$, as in the statement of Theorem \ref{the:strong-impossibility}. As $A$ solves Terminating Line Transformation on $G$, it makes $G$ acyclic and by Theorem \ref{the:strong-impossibility} it must hold that it may disconnect in some executions all $G^\prime\in\cal{G}$ in a worst-case manner. Take a $G^\prime\in\cal{G}$ of size $n^\prime$ and a $G^{\dprime}\in\cal{G}$ of size $n^{\dprime}\gg n^\prime$. If we look again at the indistinguishable executions of Theorem \ref{the:strong-impossibility}, we have that there is an edge deactivation of $A$ on $G$ that leads to disconnection of all the $c^\prime$ components of $G^\prime$ and the $c^{\dprime}\gg c^\prime$ components of $G^{\dprime}$ (we just have to use that edge as the intra-component edge) and additionally at that point, the components of both instances are in identical states. From that point on, any fair execution $\alpha^{\prime}$ of $A$ on the $c^\prime$ components leads to the formation of a spanning line of length $n^\prime$ with termination. But now consider the following unfair execution $\alpha^{\dprime}$ on the $c^{\dprime}$ components: $\alpha^{\dprime}$ simulates $\alpha^{\prime}$ on $c^\prime$ of the $c^{\dprime}$ components of $G^{\dprime}$ ignoring the remaining $c^{\dprime}-c^\prime$ components. The result is that in this case, $A$ forms a line of length only $n^\prime$ and terminates before taking into account the remaining $n^{\dprime}-n^\prime$ nodes, and this number of ignored nodes can be made arbitrarily large by taking an increasingly larger $G^{\dprime}$. It follows that there can be no such protocol $A$ for the Terminating Line Transformation problem.
\end{proof}

\begin{corollary} [Terminating Line Transformation Impossibility] \label{cor:line-impossibility}
If all nodes are initially identical, there is no protocol for Terminating Line Transformation.
\end{corollary}
\begin{proof}
Terminating Line Transformation requires making the active topology acyclic. But then Corollary \ref{cor:line-impossibility} implies that any protocol must in some (in fact, in infinitely many) executions break connectivity in a worst-case manner, that is the active topology may break into $\Theta(n)$ components, and then, by Lemma \ref{lem:line-impossibility}, any such protocol cannot be a correct protocol for Terminating Line Transformation.
\end{proof}

\subsection{The Common Neighbor Detection Assumption}
\label{subsec:neighbor-detection}

In light of the impossibility results of the previous section, we naturally ask whether some minimal strengthening of the model could make the problems solvable. To this end, we give to the nodes the ability to detect whether they have a neighbor in common. In particular, we assume that whenever two nodes interact, they can tell whether they have at that time a common neighbor (over active edges).
Clearly, this mechanism can be used to safely deactivate an edge in case it happens
that the two nodes are indeed part of a 3-cycle. If the two nodes are part
only of longer cycles they still cannot deactivate the edge with
certainty. Observe that the common neighbor detection mechanism is very local and easily implementable by almost any plausible system. For example, it only requires local names and at least 2-round local communication before neighborhood changes. Moreover, it is also an inherent capability of the
variation of population protocols in which the nodes interact in triples instead of pairs (see e.g. \cite{AADFP06,BBRR12}). Interestingly, we shall prove in this section that this minimal extra assumption overcomes the impossibility results both of Corollary \ref{cor:acyclicity-impossibility} and Corollary \ref{cor:line-impossibility}. In particular, both Acyclicity and Terminating Line Transformation become now solvable.

We begin with a simple protocol for Acyclicity (Protocol \ref{prot:star-transformer}) exploiting the common neighbor detection mechanism. The transitions are now of the form $\delta : Q\times Q\times \{0,1\}\times \{0,1\} \rightarrow Q\times Q\times \{0,1\}$, i.e. in every interaction there is an additional input from $\{0,1\}$ indicating whether the two interacting nodes have a neighbor in common (via active edges) at the beginning of their interaction. For example, rule $(p,p,1),1\ra (p,p,0)$ is interpreted as ``if both nodes are in state $p$, the edge between them is active, and they have a common neighbor, then (as the nodes are part of an active triangle) the edge between them can be (safely) deactivated and both nodes remain in state $p$''. The protocol solves Acyclicity (in contrast to the impossibility of Corollary \ref{cor:acyclicity-impossibility}) by converting any initial connected active topology to a spanning star without ever breaking connectivity.

\floatname{algorithm}{Protocol}
\renewcommand{\algorithmiccomment}[1]{// #1}
\begin{algorithm}[!h]
  \caption{\emph{Star-Transformer}}\label{prot:star-transformer}
  \begin{algorithmic}
    \medskip
    \State $Q=\{l,p\}$, initially all nodes are in state $l$
    \State $\delta$: 
    \begin{align*}
    (l,l,\cdot),\cdot &\ra (l,p,1)\\
    (l,p,0),\cdot &\ra (l,p,1)\\
    (p,p,1),1&\ra (p,p,0)   
    \phantom{\hspace{10cm}}
    \end{align*}
    \State \Comment {The last rule is triggered only if the two $p$s have a common neighbor}
  \end{algorithmic}
\end{algorithm}

In Protocol Star-Transformer, all nodes are initially leaders in state $l$. Whenever two leaders interact, one of them becomes $p$
(peripheral) and the edge between them is activated. Whenever an $l$ interacts with a $p$ over an inactive edge, the edge between them becomes activated. Finally, whenever two peripherals interact, they deactivate the edge between them only if they
have a neighbor in common (which by assumption can be locally detected).

\begin{proposition}
By assuming that nodes are equipped with the common neighbor detection mechanism, Protocol Star-Transformer solves Acyclicity in the setting in which all nodes are initially identical. In particular, the final acyclic active topology is always a spanning star.
\end{proposition}
\begin{proof}
As leaders are always pairwise eliminated, eventually a unique leader will remain and all other nodes will be in state $p$. The
leader will eventually become directly connected to all $p$s, via rule $(l,p,0)\ra (l,p,1)$. From that point
on, if there are still two $p$s that are connected with each other, then these $p$s 
have $l$ as a common neighbor and when they interact, rule $(p,p,1),1\ra (p,p,0)$
is applied and deactivates the edge between them. Thus, eventually all edges between $p$s
are deactivated and the network stabilizes to a star with $l$ at its center.
It remains to show that the active network is never disconnected. This is
ensured by the fact that the protocol only eliminates edges that belong to 3-cycles at the time of
elimination, and this cannot lead to disconnection.
\end{proof}

We now exploit the common neighbor detection assumption and the Star-Transformer protocol to give a correct and efficient protocol for the Terminating Line Transformation problem. The protocol, called Line-Transformer, assumes (as did the protocols of Section \ref{sec:unique-leader}) the ability to detect whether the local active degree of a node is equal to 1 or 2.\\

\noindent\textbf{Protocol Line-Transformer.} We give here a high-level description. All nodes are initially leaders in state $l$. When two leaders interact, one of them becomes a peripheral in state $p$ and the edge is activated. Every leader is connected to all $p$s that it encounters. Two $p$s deactivate an active edge joining them only if at the time of interaction they have a neighbor in common. When a peripheral has active local degree equal to 1, its local state is $p_1$, otherwise it is $p$ or one of some other states that we will describe in the sequel.

When a leader first sees one of its own $p_1$s (i.e. via an active edge), it initiates the formation of a line over its $p_1$ peripherals (observe that the set of $p_1$ peripherals of a leader does not remain static, as e.g. a $p_1$ becomes $p$ when some other leader connects to it). In particular, as in Protocol Line-Around-a-Star, the line will have as its ``left'' endpoint the center of the local star, which will be in a new state $e_l$, and it will start expanding over the available local $p_1$ peripherals over its right endpoint in state $l^\prime$.

The new center $e_l$ keeps connecting to new peripherals but it cannot become eliminated any more by other leaders. Pairwise eliminations only occur via any combination of $l$ and $l^\prime$. A local line expands over the local $p_1$s as follows. When the right endpoint $l^\prime$ encounters a $p_1$, which can occur only via an inactive edge, it expands on it only if the two nodes have a common neighbor (which can only be the center of the local star). If this is satisfied, the $l^\prime$ takes the place of the $p_1$, leaving
behind an $i$ (for ``internal node'' of the line) and the edge becomes
activated. Moreover, the center $e_l$ deactivates every active edge it has
with an $i$ but not with the first peripheral of the line (so the first
peripheral that the line uses must always be in a distinguished state $i_1$) and not
with the $l^\prime$ right endpoint (because that edge is always needed for common
neighbor detection during the next expansion).

Now, if the $l^\prime$ endpoint ever meets either another $l^\prime$ endpoint or an $l$
center then one of them becomes deactivated. When it meets an $l$ center we
can always prefer to deactivate the $l$ center because no line backtracking
is required in this case. When an $l$ center is deactivated, $l$ simply
becomes $p$ and the edge becomes activated (it can never be a $p_1$
immediately but this is minor).

The most interesting case is when an $l^\prime$ loses from another $l^\prime$. In this
case, the eliminated $l^\prime$ becomes $f$. The role of $f$ is to backtrack the whole
local line construction by simply converting one after the other every $i$
on its left to $p$ and finally converting $e_l$ again to $p$ (again it
cannot be a $p_1$ at the time of conversion). This backtracking process
cannot fail because $f$ has always a single $i$ (or $i_1$) active neighbor,
always the one on its left, while its right neighbor is no one initially
and a $p$ in all subsequent steps, so it knows which direction to follow.
When the backtracking process ends, all the nodes of the local star are
either $p$ or $p_1$ so they can be attracted by the stars that are still alive.

The protocol terminates, when for the first time it holds that an $e_l$ has
local degree equal to 2 after its line has for the first time length at least 3 (nodes).
When this occurs, a spanning ring has been formed and $e_l$ can deactivate
the edge $(e_l,l^\prime)$ between the two endpoints to make it a spanning line. This completes the description of the protocol. An example execution is depicted in Figure \ref{fig:line-transformer}.

\begin{figure}[!hbtp]
   \centering{
        \subfigure[]{
        \includegraphics[width=0.3\textwidth]{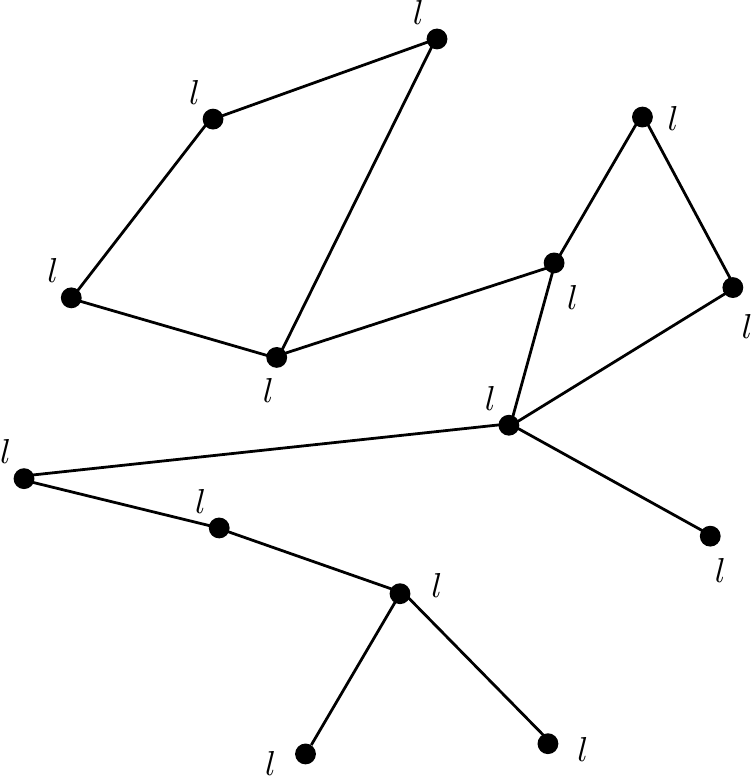}
        \label{fig:line-transformer1}}
	\hspace{1cm}
        \subfigure[]{
        \includegraphics[width=0.3\textwidth]{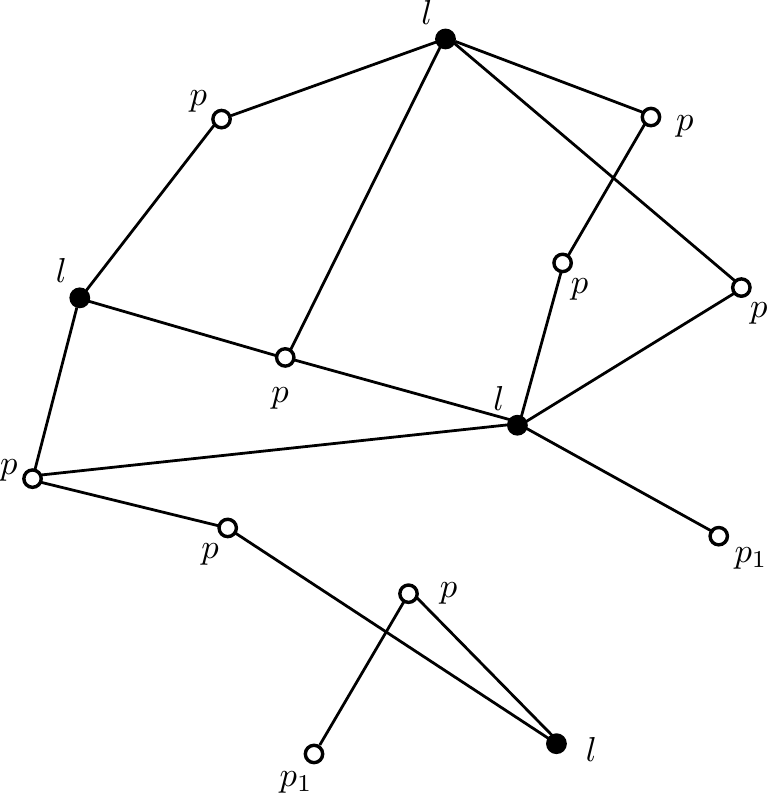}
        \label{fig:line-transformer2}}
	\hspace{1cm}
        \subfigure[]{
        \includegraphics[width=0.3\textwidth]{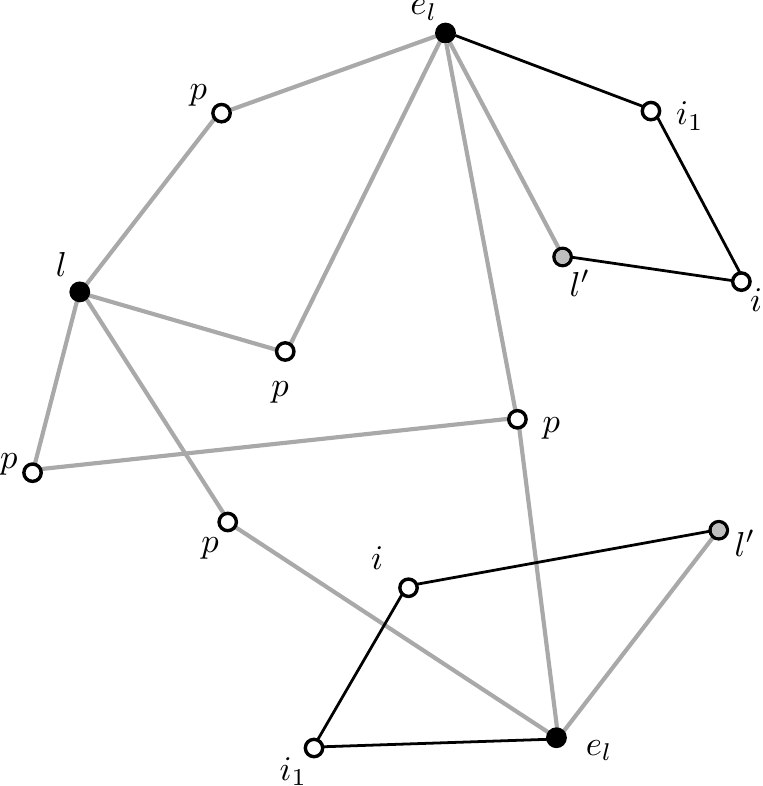}
        \label{fig:line-transformer3}}
	\hspace{1cm}
        \subfigure[]{
        \includegraphics[width=0.3\textwidth]{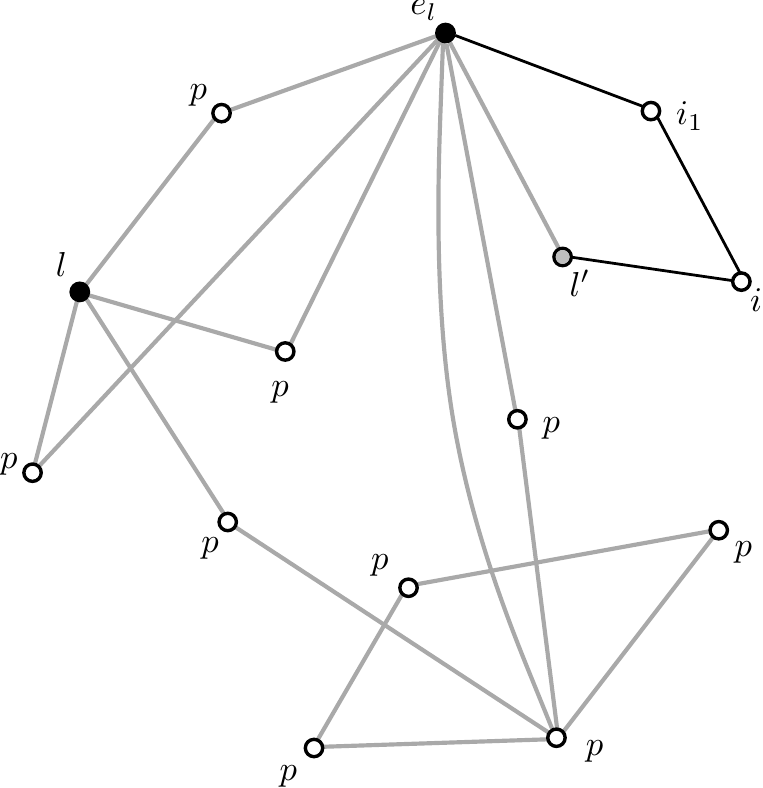}
        \label{fig:line-transformer4}} 
	\hspace{1cm}
        \subfigure[]{
        \includegraphics[width=0.3\textwidth]{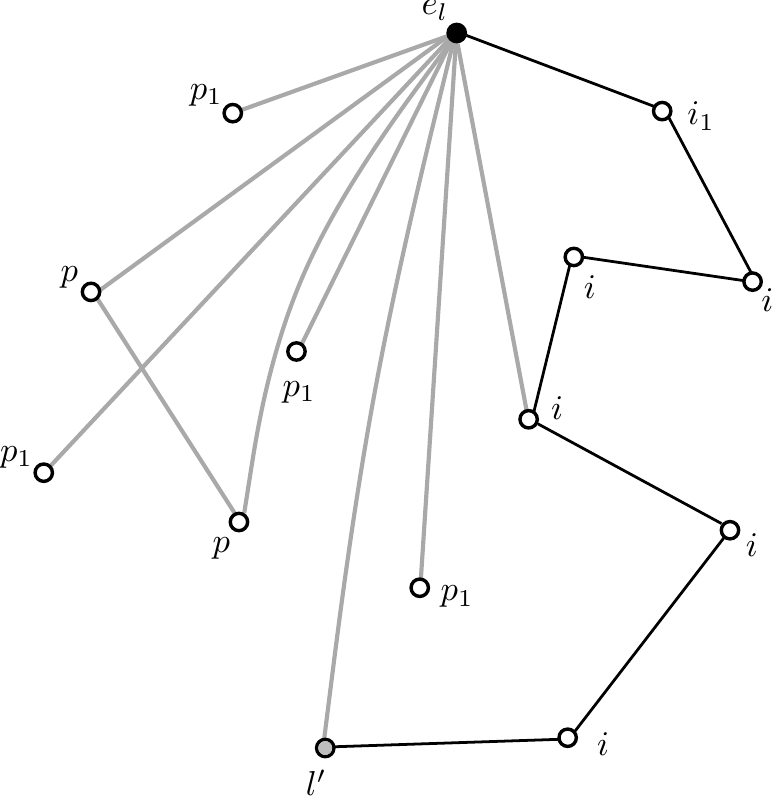}
        \label{fig:line-transformer5}}
	\hspace{1cm}
        \subfigure[]{
        \includegraphics[width=0.3\textwidth]{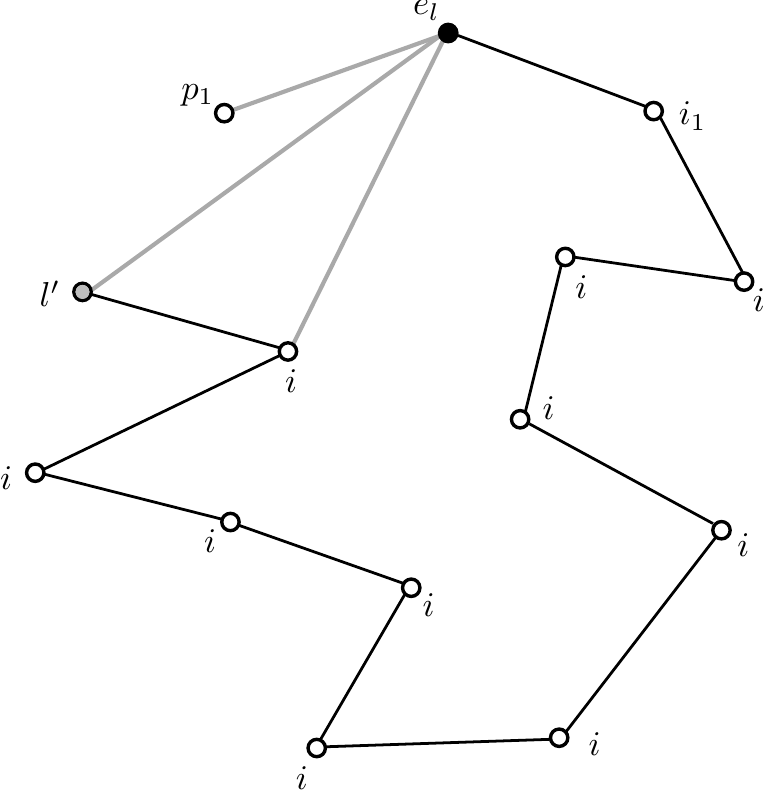}
        \label{fig:line-transformer6}}                       
        }
   \caption{An example execution of Protocol Line-Transformer. In all subfigures, black and gray edges are active and missing edges are inactive. Black and gray are used together whenever we want to highlight some subnetwork of the active network. (a) Initially, all nodes are leaders and the topology is connected. (b) Most leaders have been converted to peripherals, some leaders have attracted new peripherals, and some peripherals have disconnected from each other. (c) Two of the survived leaders have started to form lines over their $p_1$ peripherals. The centers of these stars are now in state $e_l$ (black nodes), the other enpoint of their lines is in state $l^\prime$ (gray nodes), and the lines are drawn by black edges. (d) The $l^\prime$ endpoints of the two previous lines interacted and one of them was backtracked. (e) A single line has remained. (f) The line is almost spanning.} \label{fig:line-transformer}
\end{figure}

\begin{theorem} \label{the:line-transformer}
By assuming that nodes are equipped with a common neighbor detection mechanism and have the ability to detect local degrees 1 and 2, Protocol Line-Transformer solves the Terminating Line Transformation problem in the setting in which all nodes are initially identical. Its running time is $O(n^3)$.
\end{theorem}
\begin{proof}
We begin with correctness. The network is never disconnected because the only edge deactivations
that take place are (1) between two $p$s that have a common neighbor, i.e.
breaking a 3-cycle, and (2) between an $e_l$ and an $i$. (1) clearly cannot
cause disconnection. (2) cannot cause disconnection because $e_l$ and $i$ are
still connected via the active line that begins from $e_l$ and passes over
all local $i$. Moreover, eventually a unique leader will remain because as long as there are two
stars their leaders can interact and one of them will become eliminated.

All nodes in the star of the final leader must be $p$ and $p_1$ and all other
nodes must be again $p$, $p_1$ or part of a line that is being backtracked. As
discussed above, a line that is being backtracked is guaranteed to convert
eventually all its nodes to $p$. So, eventually the unique star with an
alive leader has all remaining nodes in $p$, $p_1$ or part of its active line.
The $e_l$ center of this final star keeps connecting to $p$ and $p_1$ so eventually
all of them will have obtained an active connection to $e_l$. This implies
that from that point on, two $p$s cannot forever remain connected, because
eventually they will become connected to $e_l$ and when they will interact
with each other they will detect their common neighbor and will deactivate
the edge joining them. This implies that eventually (in the final star)
all $p$s will become $p_1$. So, the final ``star'' will actually look like a star
and a cycle beginning from $e_l$ covering an $i_1$ then several $i$ and ending at
$e_l$. When this cycle becomes spanning, this will be the first time (after
the cycle had length at least 3) that the local degree of $e_l$ becomes 2
(as there are no more $p$s and $p_1$s left in its neighborhood) and the protocol
terminates.

For the running time, observe that every star always has a leader node so the pairwise leader elimination process is never delayed. Therefore the time for a unique leader to remain is $\Theta(n^2)$. When this occurs, all other nodes are either peripherals or part of active lines with an $f$ endpoint. Eventually, all these $f$s will backtrack their lines, converting all nodes to peripherals. From that point on, the time to termination is upper bounded by the running time of Protocol Line-Around-a-Star, i.e. by $\Theta(n^2\log n)$. So, it only remains to estimate the worst-case expected time for the $f$s to backtrack their lines. Clearly, in the worst case there might be a single such line of length $\Theta(n)$. In this case, it takes $O(n^2)$ expected time for the single $f$ to interact with its next neighbor on the line, thus it takes $O(n^3)$ expected time to traverse the whole line and convert all its nodes to $p$. This part of the process dominates the running time of the protocol, therefore we conclude that the running time is $O(n^3)$ (still there is a possibility that it is a little faster if a better analysis could show that when a unique leader remains all $f$-lines are small enough, e.g. of length $o(n)$). 
\end{proof}

Table \ref{tab:line-transformation-protocols} summarizes all protocols that we developed for the Terminating Line Transformation problem, both for the case of a pre-elected unique leader (Protocols Online-Cycle-Elimination and Line-Around-a-Star in Section \ref{sec:unique-leader}) and for the case of identical nodes (Protocol Line-Transformer in the present section).

\begin{table*}[!hbtp]
\normalsize
\setlength{\tabcolsep}{10pt}
\begin{center}
\begin{tabular}{  l  c c c c c  }
  \hline
  Protocol & Leader & DD & CND & Expected Time & Lower Bound \\ \hline
  Online-Cycle-Elimination & Yes & 1 & No & $\Theta(n^4)$ & $\Omega(n^2\log n)$ \\ 
  Line-Around-a-Star & Yes & 1 & No & $\Theta(n^2\log n)$ (opt) & $\Omega(n^2\log n)$ \\
  Line-Transformer & No & 1,2 & Yes& $O(n^3)$ & $\Omega(n^2\log n)$ \\ \hline
\end{tabular}
\end{center}
\caption{All protocols developed in this work for the Terminating Line Transformation problem. For each of these protocols, the table shows whether it makes use of a pre-elected unique leader, what local degree detection it uses (DD), whether it uses common neighbor detection (CND), and also its expected running time under the uniform random scheduler. The last column shows the best known lower bound for the problem.}
\label{tab:line-transformation-protocols}
\end{table*}

Finally, we show how the spanning line formed with termination by \emph{Line-Transformer} can be used to establish that the class of computable predicates is the maximum that one can hope for in this family of models.

\begin{theorem} [Full Computational Power]
Let the initial active topology be connected, all nodes be initially identical, and let the nodes be equipped with degree in $\{1,2\}$ detection and common neighbor detection. Then for every predicate $p\in\rem{SSPACE}(n^2)$ there is a terminating NET that computes $p$.  
\end{theorem}
\begin{proof}
Let $N$ be the deterministic TM that decides $p$ in space $O(n^2)$. By Theorem \ref{the:line-transformer} there is a protocol, namely Line-Transformer, that solves the Terminating Line Transformation problem. In particular, when the protocol terminates, the active topology is a spanning line and has a unique leader. We sequentially compose the Line-Transformer protocol with another protocol, called Simulator, which transforms the spanning line to some other convenient topology and uses that topology in order to simulate $N$. Recall that, as usual, the predicate $p$ is defined on input assignments to the nodes from some alphabet $X$. We assume that the nodes maintain their inputs throughout the execution of Line-Transformer so that their original inputs are available when the Simulator begins. We remind that the restriction to the \emph{symmetric} subclass of $\rem{SPACE}$ is because the active clique is a permissible initial topology and the nodes are initially indistinguishable so, in general, it is meaningless to consider ordered inputs in this setting (still, if desired, we could have a non input-symmetric TM decide the ordered input of the constructed spanning line, which taken over all possible executions will produce all possible ordered inputs).

We now describe the protocol Simulator, which is essentially an adaptation of the constructors of Theorems 11 and 12 of \cite{MS14}. Given the spanning line with a unique leader provided with termination by Line-Transformer, the protocol partitions the population into two equal sets $U$ and $M$ such that (i) there is an active perfect matching between $U$ and $M$, (ii) the nodes of $U$ are arranged in a spanning line which has its endpoints in state $q_1$, the internal nodes in state $q_2$, has additionally a unique leader on some node, and all inputs of the population have been transferred to the nodes of $U$ (e.g. each node of $U$ also remembers the input of a distinct node from $M$), (iii) all nodes in $M$ are in state $q_m$ and all edges between nodes of $M$ are inactive. The above transformations are straightforward to achieve. In case of an even number of nodes, the leader first counts the $n/2$ leftmost nodes on the spanning line (by simulating a TM that does so on the nodes of the spanning line, which play the role of a linear memory). Then the leader appends the input of node $n/2+i$, for $1\leq i\leq n/2$, to the state of node $i$, so that the $n/2$ leftmost nodes know all inputs (each node recording two of them). Then the leader deactivates all edges to the right of node $n/2$ which will leave a line of length $n/2$ with a unique leader and $n/2$ isolated nodes. Now the leader waits to match each node of the remaining line to a distinct node from the isolated ones. In this manner, the leader eventually forms a line of length $n/2$ (which is the set $U$ mentioned above) matched to $n/2$ nodes that have no edges between them (the set $M$ above). It only remains for the leader to visit one after the other all nodes of the new construction in order to set their states as required above. The case of odd $n$ is essentially equivalent. The only difference is that one of the nodes in $U$ will have to remember the input of a redundant node (i.e. three nodes in total) which does not participate in the matching and forever remains isolated.

The purpose of set $M$ is to constitute the $\Theta(n^2)$ memory required by the TM $N$. The goal is to exploit the $(n/2)(n/2-1)/2$ edges of set $M$ as the binary cells of the simulated TM (see Figure \ref{fig:simulator} for an illustration). The line of set $U$ controls the simulation and provides a unique ordering of the edges between nodes in $M$. The idea is to simulate the head of $N$ by two tokens that move on the line of $U$. The positions $u_i,u_j$ of these tokens on the line uniquely identify an edge of $M$, which is the one between the matched nodes $v_i,v_j$ from $M$. When the $N$ wants to move right (left) the rightmost token is moved one step to the right (left, resp.) and if it has already reached the right endpoint (left token, resp.) then the leftmost token is moved one step to the right and the rightmost token moves to the right neighbor of the leftmost token (the opposite, resp.). Whenever $N$ wants to modify the bit of its current cell, the protocol marks by a special activating or deactivating state the $M$-nodes matching the current positions of the $U$-tokens. Then an interaction between two such marked $M$-nodes activates or deactivates, respectively, the edge between them. Reading the state of a cell is performed similarly.
\end{proof}

\begin{figure}[!hbtp]
\centering{
\includegraphics[width=0.75\textwidth]{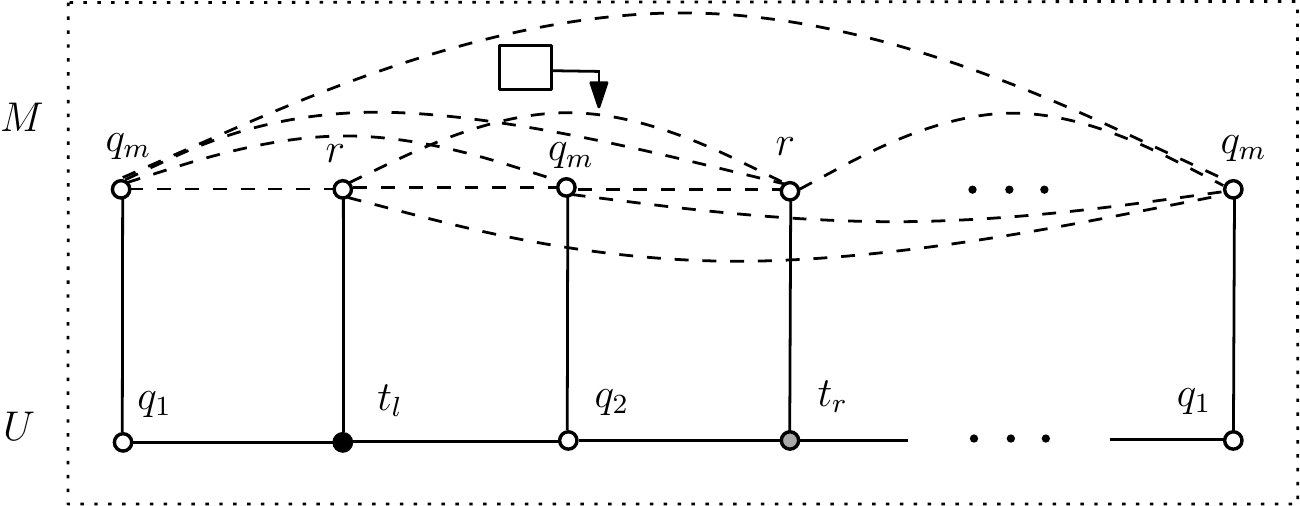}
}
\caption{The spanning line of Line-Transformer has been transformed by the Simulator to the topology depicted here. The spanning line of set $U$ controls the simulation via its two tokens (here placed on the black and gray nodes). Their matching nodes from $M$ are here marked $r$ which means that a ``reading'' of the cell simulated by the edge between the two nodes is being performed.} \label{fig:simulator}
\end{figure}

\begin{remark}
At this point the reader might justifiably think that interchanging the role of edge-states 0 and 1 would be sufficient for establishing the same result even in the case in which all edges are initially inactive. By this reasoning, the protocol should be able to construct a spanning \emph{non-line} (i.e. one in which all edges are in state 0) and all other edges of the system are active. In fact, for one moment this might seem to contradict the whole reasoning and development of the present paper, because it would imply that no initial active connectivity is required for termination as it could just be replaced by an equivalent inactive connectivity. A more careful look reveals that though this is formally plausible it still leads to unrealistic assumptions. The reason is that this is valid only if a node is able to detect small local \emph{non-degrees}, for example whether it has precisely one \emph{non-neighbor} in the system. In a real system, where active edges are expected to correspond to established communication links and inactive to non-existing communication links, such information is \emph{highly non-local} and provides a node with an \emph{implicit knowledge of the size of the system}. The same would hold for a common non-neighbor mechanism. On the other hand, the corresponding assumptions based on active (i.e. established) connections are totally local and expected to hold in almost any distributed system.
\end{remark}

\section{Conclusions and Further Research}
\label{sec:conclusions}

There are many open problems related to the findings of the present work. We have shown that initial connectivity of the active topology combined with the ability of the protocol to transform the topology yield, under some additional minimal and local assumptions, an extremely powerful model. We managed to show this by developing protocols that transform the initial topology to a convenient one (in our case the spanning line) while always \emph{preserving the connectivity of the topology}. Though arbitrary connectivity breaking makes termination impossible, still we have not excluded the possibility that some protocol performs some ``controlled'' connectivity breaking during its course, being always able to correctly reassemble the disconnected parts and terminate. For example, imagine a protocol with a unique leader that performs only a bounded number of simultaneous disconnections and always remembers the number of released components. It would be interesting to know to what extent is this possible and whether it is sensitive to the availability of a unique leader.

Another issue has to do with the underlying interaction model. Throughout this work we have assumed that the underlying interaction graph is the clique $K_n$ and all of our protocols largely exploit this. Though this model is a convenient starting point to understand the basic principles of algorithmic transformations of networks, it is obvious that it is highly non-local. Realistic implementations would probably require more local or geometrically constrained models (like the one of \cite{Mi15}), for example, one in which, at any given time, a node can only communicate with nodes at active distance at most 2. This could be even further restricted by assuming that nodes have bounded active degree (e.g. at most 4) which wouldn't allow a node gradually increase its connectivity by increasing its local degree. Instead a node should repetitively drop some of its active edges in order to make new connections and move on the active topology. It would be of both practical and theoretical importance to develop such models and give protocols and characterizations for them. It is also valuable to consider the Terminating Line Transformation and Acyclicity problems in models of computationally weak (and probably also anonymous) robots moving in the plane. Imagine, for example, a swarm of such robots arranged initially in a spanning ring, which is very symmetric and may not allow for non-trivial computations, with the goal being to transform their arrangement to a spanning line or some other topology that allows for global computations.

There are also some more technical intriguing open questions. The most prominent one is whether protocol Line-Transformer is time-optimal. Recall that its running time was shown to be $O(n^3)$. First of all, it is not clear whether the analysis is tight. The subroutine that dominates the running time is the one that tries to form a spanning line over the peripheral nodes, which is restricted by the fact that the partial lines of ``sleeping'' stars have to either be backtracked (which is what our solution does) or merged somehow with the lines of ``awake'' stars. We should mention that the spanning line subroutine that backtracks many ``sleeping'' lines in parallel is an immediate improvement of the best spanning line protocol of \cite{MS14}, called Fast-Global-Line. The improvement is due to the fact that instead of having the awake leader backtrack node-by-node sleeping lines, we now have any sleeping line backtrack itself, so that many backtrackings occur in parallel. We also have experimental evidence showing a small improvement \cite{ALMS15} but still we do not have a proof of whether this is also an asymptotic improvement. For example, is it the case that the running time of this improvement is $O(n^3/\log n)$ (or even smaller)? This question is open. There is also room for lower bounds. Apart from the obvious lower bound for the Terminating Line Transformation problem with identical nodes, one could also focus on the spanning line construction problem with initially disconnected nodes (i.e. the Spanning Line problem of \cite{MS14}). The reason is that an improvement to this problem would probably imply an improvement for Terminating Line Transformation by using the protocol as an improved subroutine of Line-Transformer for forming the lines over the peripherals of the star. The best lower bound known for Spanning Line is $\Omega(n^2)$. Some first attempts suggest that it might be non-trivial to improve this to $\Omega(n^2\log n)$. 

Another potential direction is to study the Terminating Line Transformation and other transformation problems in restricted families of input topologies. For example, if the initial topology is guaranteed to be a ring, then, even if all nodes are initially identical (i.e. no unique leader), we can transform the ring into a spanning line with a single edge deactivation (and no activation throughout the course of the protocol). Essentially, every node issues a trace expanding either clockwise or counterclockwise and we ensure that only one of the traces manages eventually to make a complete turn and reach its own tail. That trace is an elected leader-trace and can safely deactivate the edge joining its head to its tail (observe that more than one deactivations by distinct traces could disconnect the ring). A more interesting case that is worth being studied is the one in which the initial topology is always a spanning tree. Moreover, the observation that the running time of protocols seems to depend on the number of edges that have to be modified, suggests to study the running time of protocols as a function of the number of (active) edges of the initial topology or as the difference in the number of edges between the initial and the target topology.

Finally, it would be interesting to investigate the role of UIDs (or even non-unique identities, as in \cite{BCR15}) in network construction and network transformations. One part of the question concerns the fact that UIDs seem to help in improving efficiency. The other part concerns the fact that UIDs give rise to a much richer set of network transformation problems as they now involve the substantially larger class of \emph{node-labeled graphs}.\\

\noindent \textbf{Acknowledgements.} We would like to thank Leszek Gasieniec for bringing to our attention the importance of the network reconfiguration problem and also the anonymous reviewers of this article, whose thorough comments have helped us to improve our work substantially.



\newcommand{\etalchar}[1]{$^{#1}$}

%
%

\end{document}